\documentclass[11pt]{article}
\usepackage[margin=1in]{geometry}

\usepackage[utf8]{inputenc} 
\usepackage{booktabs}       
\usepackage{amsfonts}       
\usepackage{nicefrac}       
\usepackage{microtype}      
\usepackage{xcolor}         
\usepackage{parskip}        
\usepackage{amsmath, amsthm, mathtools, dsfont}
\usepackage{txfonts} 
\usepackage{color,graphicx}
\usepackage{url,hyperref}
\usepackage{enumerate}
\usepackage{algorithm}
\usepackage[noend]{algpseudocode}
\usepackage{fancybox}
\usepackage{xspace}
\usepackage{cuted,tcolorbox,lipsum}
\usepackage{multicol}
\usepackage{xcolor}
\usepackage[T1]{fontenc}
\usepackage{tikz}
\usepackage{tcolorbox}
\usepackage{hyperref}
\usepackage[all]{hypcap}
\usetikzlibrary{calc}
\tcbuselibrary{skins}
\usepackage[colorinlistoftodos]{todonotes}
\usepackage{array,multirow}

\usepackage{amsthm}
\usepackage{framed}
\usepackage{mdframed}

\newcommand*\samethanks[1][\value{footnote}]{\footnotemark[#1]}

\theoremstyle{definition}

\colorlet{shadecolor}{gray!20}

\newtheorem{theorem}{Theorem}
\newtheorem{corollary}[theorem]{Corollary}
\newtheorem{lemma}[theorem]{Lemma}

\newtheorem{definition}[theorem]{Definition}

\newcommand{\mf}{\ensuremath{f^*}}

\newcommand{\coloring}{\textit{coloring}}

\newcommand{\staticcoloring}{\textit{static coloring}}

\newcommand{\validstaticcoloring}{\textit{valid static coloring}}

\newcommand{\dynamiccoloring}{\textit{dynamic coloring}}

\newcommand{\minimaldynamiccoloring}{\textit{minimal dynamic coloring}}

\newcommand{\pairconst}{\textit{pair constraint}}

\newcommand{\pcsfthree}{{\textsc{PCSF3}}}

\newcommand{\rpcsf}{{\textsc{IPCSF}}}

\newcommand{\finddeltap}{{\textsc{FindDeltaP}}}

\newcommand{\finddeltae}{{\textsc{FindDeltaE}}}

\newcommand{\reducetightpairs}{{\textsc{ReduceTightPairs}}}

\newcommand{\setpairgraph}{{\textsc{SetPairGraph}}}

\newcommand{\checksetistight}{{\textsc{CheckSetIsTight}}}

\newcommand{\mincut}{{\textsc{MaxFlow}}}

\newcommand{\deltaS}{\ensuremath{\delta(S)}}

\newcommand{\G}{\ensuremath{G}}

\newcommand{\V}{\ensuremath{V}}

\newcommand{\E}{\ensuremath{E}}

\newcommand{\y}{\ensuremath{y}}

\newcommand{\ys}{\ensuremath{\y_S}}

\newcommand{\yij}{\ensuremath{\y_{ij}}}

\newcommand{\ysij}{\ensuremath{\y_{Sij}}}

\newcommand{\cc}{\ensuremath{c}}

\newcommand{\ce}{\ensuremath{
\cc_e}}

\newcommand{\pij}{\ensuremath{\pi_{ij}}}

\newcommand{\demands}{\ensuremath{\mathcal{D}}}

\newcommand{\currentsets}{\ensuremath{FC}}

\newcommand{\activesets}{\ensuremath{ActS}}

\newcommand{\Deltae}{\ensuremath{\Delta_e}}

\newcommand{\Deltap}{\ensuremath{\Delta_p}}

\newcommand{\Deltat}{\ensuremath{\Delta}}

\newcommand{\F}{\ensuremath{F}}

\newcommand{\Fp}{\ensuremath{F'}}

\newcommand{\Q}{\ensuremath{Q}}

\newcommand{\GF}{\ensuremath{\mathcal{G}}}

\newcommand{\source}{\ensuremath{source}}

\newcommand{\sink}{\ensuremath{sink}}

\newcommand{\SC}{\ensuremath{C_{\source}}}

\newcommand{\MC}{\ensuremath{C_{min}}}

\newcommand{\f}{\ensuremath{f}}

\newcommand{\mfout}{\ensuremath{(\mf, \MC, \f)}}

\newcommand{\I}{\ensuremath{I}}

\newcommand{\R}{\ensuremath{R}}

\newcommand{\A}{\ensuremath{\mathcal{CC}}}

\newcommand{\B}{\ensuremath{\mathcal{CP}}}

\newcommand{\C}{\ensuremath{\mathcal{PC}}}

\newcommand{\D}{\ensuremath{\mathcal{PP}}}

\newcommand{\va}{\ensuremath{cc}}

\newcommand{\vb}{\ensuremath{cp}}

\newcommand{\vba}{\ensuremath{cp_1}}

\newcommand{\vbb}{\ensuremath{cp_2}}

\newcommand{\vc}{\ensuremath{pc}}

\newcommand{\vd}{\ensuremath{pp}}

\newcommand{\OPT}{\ensuremath{OPT}}

\newcommand{\sol}{\text{SOL}}

\newcommand{\fopt}{\ensuremath{F^*}}

\newcommand{\qopt}{\ensuremath{Q^*}}

\newcommand{\foptr}{\ensuremath{F^*_R}}

\newcommand{\qoptr}{\ensuremath{Q^*_R}}

\newcommand{\foptp}{\ensuremath{F'_R}}

\newcommand{\qoptp}{\ensuremath{Q'_R}}

\newcommand{\dopt}{\ensuremath{d_{\fopt}}}

\newcommand{\OPTR}{\ensuremath{OPT_R}}

\newcommand{\costalg}{\ensuremath{cost(\rpcsf)}}

\newcommand{\comma}{,}
\title{2-Approximation for Prize-Collecting Steiner Forest}
\date{}

\author{
Ali Ahmadi\thanks{University of Maryland.}\\
\texttt{ahmadia@umd.edu}
\and
Iman Gholami\samethanks\\
\texttt{igholami@umd.edu}
\and
MohammadTaghi Hajiaghayi\samethanks\\
\texttt{hajiagha@umd.edu}
\and 
Peyman Jabbarzade\samethanks\\
\texttt{peymanj@umd.edu}
\and
Mohammad Mahdavi\samethanks\\
\texttt{mahdavi@umd.edu}
}

\begin{document}
\maketitle


\begin{abstract}
Approximation algorithms for the prize-collecting Steiner forest problem (PCSF) have been a subject of research for over three decades, starting with the seminal works of Agrawal, Klein, and Ravi~\cite{AKRSTOC91,DBLP:journals/siamcomp/AgrawalKR95}  and Goemans and Williamson \cite{GWSODA92,DBLP:journals/siamcomp/GoemansW95} on Steiner forest and prize-collecting problems.
In this paper, we propose and analyze a natural deterministic  algorithm for PCSF that achieves a $2$-approximate solution in polynomial time. 
This represents a significant improvement compared to the previously best known algorithm with a $2.54$-approximation factor developed by Hajiaghayi and Jain \cite{DBLP:conf/soda/HajiaghayiJ06} in 2006.
Furthermore, K{\"{o}}nemann, Olver, Pashkovich, Ravi, Swamy, and Vygen~\cite{DBLP:conf/approx/KonemannOP0SV17} have established an integrality gap of at least $9/4$ for the natural LP relaxation for PCSF.
However, we surpass this gap through the utilization of a combinatorial algorithm and a novel analysis technique.
Since $2$ is the best known approximation guarantee for Steiner forest problem \cite{DBLP:journals/siamcomp/AgrawalKR95} (see also~\cite{DBLP:journals/siamcomp/GoemansW95}), which is a special case of PCSF, our result matches this factor and closes the gap between the Steiner forest problem and its generalized version, PCSF.
\end{abstract}


\section{Introduction}

The Steiner forest problem, also known as the generalized Steiner tree problem, is a fundamental NP-hard problem in computer science and a more general version of the Steiner tree problem. 
In this problem, given an undirected graph $\G=(\V, \E, \cc)$ with edge costs $\cc: \E \rightarrow \mathbb{R}_{\ge 0}$ and a set of pairs of vertices $\demands = \{(v_1, u_1), (v_2, u_2), \cdots (v_k, u_k)\}$ called demands, the objective is to find a subset of edges with the minimum total cost that connects $v_i$ to $u_i$ for every $i \le k$.
In this paper, our focus is on the prize-collecting Steiner forest problem (PCSF), which is a generalized version of the Steiner forest problem.

Balas~\cite{Balas89} first introduced general prize-collecting problems in 1989 and Bienstock, Goemans,   Simchi-Levi, and Williamson~\cite{BienstockGSW93} developed the first approximation algorithms for these problems. 
In the prize-collecting version of the Steiner forest problem, we are given an undirected graph $\G=(\V, \E, \cc)$ with edge costs $\cc: \E \rightarrow \mathbb{R}_{\ge 0}$ and a set of pairs of vertices $\demands = \{(v_1, u_1), (v_2, u_2), \cdots (v_k, u_k)\}$ called demands, along with non-negative penalties $\pij$ for each demand $(i, j)$. 
The objective is to find a subset of edges and pay their costs, while also paying penalties for the demands that are not connected in the resulting forest. 
Specifically, we aim to find a subset of demands $\Q$ and a forest $\F$ such that if a demand $(i, j)$ is not in $\Q$, its endpoints $i$ and $j$ are connected in $\F$, while minimizing the total penalty of the demands in $\Q$ and the sum of the costs of the edges in $\F$.
Without loss of generality, we assign a penalty of $0$ to pairs that do not represent a demand, ensuring that there is a penalty associated with each pair of vertices.
This allows us to define the penalty function $\pi : \V \times \V \rightarrow \mathbb{R}_{\ge 0}$, where $\V \times \V$ represents the set of all unordered pairs of vertices with $i \neq j$.
In this paper, we significantly improve the approximation factor of the best-known algorithm for PCSF.

For the Steiner forest problem, the first approximation algorithm was introduced by Agrawal, Klein, and Ravi \cite{DBLP:journals/siamcomp/AgrawalKR95}.
Their algorithm addressed a more generalized version of the Steiner forest problem and achieved a $2$-approximation for Steiner forest.
Later, Goemans and Williamson \cite{DBLP:journals/siamcomp/GoemansW95} provided a simplified simulation of their algorithm, which yields a $(2-\frac{2}{n})$-approximate solution for the Steiner forest problem, where $n$ is the number of vertices\footnote{Indeed
Goemans and Williamson~\cite{Hochbaum96}(Sec 4.6.1) explicitly mention
 ``... the primal-dual algorithm we have presented  simulates an algorithm of Agrawal, Klein, and Ravi [AKR95]. Their algorithm was the first approximation algorithm for this [Steiner forest a.k.a. generalized Steiner tree] problem and has motivated much of the authors’ research in this area.''; the seminal work of Agrawal, Klein, and Ravi~\cite{AKRSTOC91,DBLP:journals/siamcomp/AgrawalKR95} recently received {\em The 30-year STOC Test-of Time Award}.}.
However, no further advancements have been made in improving the approximation factor of this problem since then. 
There has been a study focused on analyzing a natural algorithm for the problem, resulting in a constant approximation factor worse than $2$ \cite{DBLP:conf/stoc/Gupta015}.
In this paper, we close the gap between the Steiner forest problem and its generalized version, PCSF, by presenting a $2$-approximation algorithm for PCSF.

The Steiner tree problem is a well-studied special case of the Steiner forest problem.
In the Steiner tree problem, one endpoint of every demand is a specific vertex known as $root$.
In contrast to the Steiner forest problem, the approximation factor of the Steiner tree problem has seen significant progress since the introduction of the $(2-\frac{2}{n})$-approximation algorithm by Goemans and Williamson \cite{DBLP:journals/siamcomp/GoemansW95}. 
Several improvements have been made \cite{DBLP:journals/algorithmica/Zelikovsky93, DBLP:journals/siamdm/RobinsZ05, DBLP:journals/jco/KarpinskiZ97}, leading to a $1.39$ approximation factor achieved by Byrka, Grandoni, Rothvo{\ss}, and Sanit{\`{a}}  \cite{DBLP:conf/stoc/ByrkaGRS10}.
Lower bounds have also been established, with \cite{DBLP:conf/coco/Karp72} proving the NP-hardness of the Steiner tree problem and consequently the Steiner forest problem, and \cite{DBLP:journals/tcs/ChlebikC08, DBLP:journals/ipl/BernP89} demonstrating that achieving an approximation factor within $96/95$ is NP-hard.
These advancements, along with the established lower bounds, underscore the extensive research conducted in the field of Steiner tree and Steiner forest problems.

Regarding the previous works in the prize-collecting version of these problems, Goemans and Williamson \cite{DBLP:journals/siamcomp/GoemansW95} provided a $(2-\frac{1}{n-1})$-approximation algorithm for  prize-collecting Steiner tree (PCST) and prize-collecting TSP problem (PCTSP) in addition to their work on the Steiner forest problem.
However, they did not provide an algorithm specifically for the PCSF problem, leaving it as an open problem.
Later, Hajiaghayi and Jain~\cite{DBLP:conf/soda/HajiaghayiJ06} in 2006 proposed a deterministic primal-dual $(3-\frac{2}{n})$-approximation algorithm for the PCSF problem, which inspired our work. 
They also presented a randomized LP-rounding $2.54$-approximation algorithm for the problem. 
In their paper, they mentioned that finding a better approximation factor, ideally $2$, remained an open problem. 
However, no improvements have been made to their result thus far. 
Furthermore, other $3$-approximation algorithms have been proposed using cost-sharing \cite{DBLP:conf/soda/GuptaKLRS07} or iterative rounding \cite{DBLP:conf/latin/HajiaghayiN10} (see e.g.~\cite{BateniH12,HajiaghayiKKN12,SharmaSW07} for further work on PCSF and its generalizations). 
Our paper is the first work that improves the approximation factor of \cite{DBLP:conf/soda/HajiaghayiJ06}.

Moreover, advancements have been made in the PCST problem since the initial $(2-\frac{1}{n-1})$-approximation algorithm by Goemans and Williamson \cite{DBLP:journals/siamcomp/GoemansW95}.
Archer, Bateni, Hajiaghayi, and Karloff \cite{DBLP:journals/siamcomp/ArcherBHK11} presented a $1.9672$-approximation algorithm for PCST, surpassing the barrier of a $2$-approximation factor.
Additionally, there have been significant advancements in the prize-collecting TSP, which shares similarities with the LP formulation of PCST. 
Various works have been done in this area \cite{DBLP:journals/siamcomp/ArcherBHK11, DBLP:journals/corr/abs-0910-0553,DBLP:conf/stoc/BlauthN23}, and the currently best-known approximation factor is $1.599$ \cite{DBLP:journals/corr/abs-2308-06254}.
These works demonstrate the importance and interest surrounding prize-collecting problems, emphasizing their significance in the research community.

For a while, the best-known lower bound for the integrality gap of the natural LP relaxation for PCSF was $2$. 
However, K{\"{o}}nemann, Olver, Pashkovich, Ravi, Swamy, and Vygen \cite{DBLP:conf/approx/KonemannOP0SV17} proved that the integrality gap of this LP is at least $9/4$. 
This result suggests that it is not possible to achieve a $2$-approximation algorithm for PCSF solely through primal-dual approaches based on the natural LP, similar to the approaches presented in \cite{DBLP:conf/soda/HajiaghayiJ06, DBLP:conf/latin/HajiaghayiN10}. 
This raises doubts about the possibility of achieving an algorithm with an approximation factor better than $9/4$.

However, in this paper, we provide a positive answer to this question. Our main result, Theorem $\ref{thm:main_theorem}$, demonstrates the existence of a natural deterministic algorithm for the PCSF problem that achieves a $2$-approximate solution in polynomial time.

\begin{theorem}
\label{thm:main_theorem}
There exists a deterministic algorithm for the prize-collecting Steiner forest problem that achieves a $2$-approximate solution in polynomial time.
\end{theorem}

We address the $9/4$ integrality gap by analyzing a natural iterative algorithm. 
In contrast to previous approaches in the Steiner forest and PCSF fields that compare solutions with feasible dual LP solutions, we compare our solution directly with the optimal solution and assess how much the optimal solution surpasses the dual. 
It is worth noting that our paper does not rely on the primal and dual LP formulations of the Steiner forest problem. 
Instead, we employ a coloring schema that shares similarities with primal-dual approaches. 
While LP techniques could be applied to various parts of our paper, we believe that solely relying on LP would not be sufficient, particularly when it comes to overcoming the integrality gap.
Furthermore, although coloring has been used in solving Steiner problems \cite{DBLP:journals/jacm/BateniHM11}, our approach goes further by incorporating two interdependent colorings, making it novel and more advanced.

In addition, we analyze a general approach that can be applied to various prize-collecting problems.
In any prize-collecting problem, an algorithm needs to make decisions regarding which demands to pay penalties for and which demands to satisfy.
Let us assume that for a prize-collecting problem, we have a base algorithm $A$.
We propose a natural iterative algorithm that begins by running $A$ on an initial instance and storing its solution as one of the options for the final solution.
The solution generated by algorithm $A$ pays penalties for some demands and satisfies others.
Subsequently, we assume that all subsequent solutions generated by our algorithm will pay penalties for the demands that $A$ paid, set the penalties of these demands to zero, and run $A$ again on the modified instance.
We repeat this procedure recursively until we reach a state where algorithm $A$ satisfies every demand with a non-zero penalty, meaning that further iterations will yield the same solution. 
This state is guaranteed to be reached since the number of non-zero demands decreases at each step.
Finally, we obtain multiple solutions for the initial instance and select the one with the minimum cost.
This natural iterative algorithm could be effective in solving prize-collecting problems, and in this paper, we analyze its application to the PCSF problem using a variation of the algorithm proposed in \cite{DBLP:conf/soda/HajiaghayiJ06} as our base algorithm.

One interesting aspect of our findings is that the current best algorithm for the Steiner forest problem achieves an approximation ratio of $2$, and this approximation factor has remained unchanged for a significant period of time. 
It is worth noting that the Steiner forest problem is a specific case of PCSF, where each instance of the Steiner forest can be transformed into a PCSF instance by assigning a sufficiently large penalty to each demand.
Since our result achieves the same approximation factor for PCSF, improving the approximation factor for the PCSF problem proves to be more challenging compared to the Steiner forest problem. 
In future research, it may be more practical to focus on finding a better approximation factor for the Steiner forest problem, which has been an open question for a significant duration. 
Additionally, investigating the tightness of the $2$-approximation factor for both problems could be a valuable direction for further exploration.


\subsection{Algorithm and Techniques}

In this paper, we introduce a coloring schema that is useful in designing algorithms for Steiner forest, PCSF, and related problems. 
This coloring schema provides a different perspective from the algorithms proposed by Goemans and Williamson in \cite{DBLP:journals/siamcomp/GoemansW95} for Steiner forest and Hajiaghayi and Jain in \cite{DBLP:conf/soda/HajiaghayiJ06} for PCSF. 
In Section \ref{sec:3_apx_alg}, we provide a detailed representation of the algorithm proposed in \cite{DBLP:conf/soda/HajiaghayiJ06} using our coloring schema. 
The use of coloring enhances the intuitiveness of the algorithm, compared to the primal-dual approach utilized in \cite{DBLP:conf/soda/HajiaghayiJ06}, and enables the analysis of our 2-approximation algorithm.
Additionally, we introduce a modification to the algorithm of \cite{DBLP:conf/soda/HajiaghayiJ06}, which is essential for the analysis of our 2-approximation algorithm.
Finally, in Section \ref{sec:2_apx_alg}, we present an iterative algorithm and prove its 2-approximation guarantee for PCSF.

Here, we provide a brief explanation of how coloring intuitively solves the Steiner forest problem.
We then present a 3-approximation algorithm and subsequently a 2-approximation algorithm for PCSF.

\paragraph{Steiner forest.}
We start with an empty forest $\F$ to hold our solution.
The set $\currentsets$ represents the connected components of $\F$ at each moment. 
A connected component of $\F$ is considered an active set if it requires extension to connect with other components and satisfy the demands it cuts.
We maintain a subset of $\currentsets$ as active sets in $\activesets$.
Starting from this point, we consider each edge as a curve with its length equal to its cost.

In each iteration of our algorithm, every active set $S \in \activesets$ is assigned a distinct color, which is used to simultaneously color its cutting edges at the same speed.
The cutting edges of a set $S$ are defined as the edges that have exactly one endpoint within $S$. 
Our coloring procedure proceeds by coloring the remaining uncolored sections of these edges.
An edge is in the process of getting colored at a given moment if it connects different connected components of $\F$ and has at least one endpoint corresponding to an active set.
Additionally, if an edge is a cutting edge for two active sets, it is colored at a speed twice as fast as an edge that is a cutting edge for only one active set.
We continue this coloring process continuously until an edge $e$ is fully colored, and then we add it to the forest $\F$.
Afterwards, we update $\currentsets$ and $\activesets$ accordingly, as defined earlier.
It is important to note that we only add edges to $\F$ that connect different connected components, ensuring that $\F$ remains a forest. 
Furthermore, since the set of all connected components of $\F$ forms a laminar set over time, our coloring schema is also laminar.
Refer to Figure \ref{fig:growth} for clarity on the coloring process.

\begin{figure}
\centering
\begin{tikzpicture}[scale=0.2]
\def\big{15}
\def\small{12}
\def\r{0.5}
\foreach \nd/\deg in {a1/0, a2/30, b1/150, b2/180, c1/255, c2/285} {
    \draw[fill=black] (-\deg:\big) circle (\r);
    \node (\nd) at (-\deg: \big) {};
};
\draw[fill=black] (-160:\small) circle (\r);
\node (b3) at (-160: \small) {};
\draw[line width=3, color=orange] (a1) -- ($(a1)!0.5!(a2)$);
\draw[line width=3, color=violet] (a2) -- ($(a2)!0.5!(a1)$);
\draw[line width=3, color=purple] (c1) -- ($(c1)!0.4!(c2)$);
\draw[line width=3, color=cyan] ($(c1)!0.4!(c2)$) -- (c2);
\draw[line width=3, color=teal] (b1) -- (b3);
\draw[line width=3, color=olive] (b2) -- ($(b2)!0.6!(b3)$);
\draw[line width=3, color=teal] ($(b2)!0.6!(b3)$) -- (b3);
\draw[] (b1) -- (b2);
\draw[line width=3, color=teal] (b1) -- ($(b1)!0.4!(b2)$);
\draw[line width=3, color=olive] (b2) -- ($(b2)!0.4!(b1)$);
\foreach \x/\y/\col/\color in {b2/c1/purple/olive, b2/c2/cyan/olive, b1/a2/blue/teal} {
    \draw[] (\x) -- (\y);
    \draw[line width=3, color=red] (\x) -- ($(\x)!0.4!(\y)$);
    \draw[line width=3, color=\color] (\x) -- ($(\x)!0.2!(\y)$);
    \draw[line width=3, color=\col] (\y) -- ($(\y)!0.2!(\x)$);
}
\draw[] (a1) -- (c2);
\draw[line width=3, color=blue] (a1) -- ($(a1)!0.4!(c2)$);
\draw[line width=3, color=orange] (a1) -- ($(a1)!0.2!(c2)$);
\draw[line width=3, color=blue] (a2) -- ($(a2)!0.4!(b1)$);
\draw[line width=3, color=violet] (a2) -- ($(a2)!0.2!(b1)$);
\draw[line width=3, color=cyan] (c2) -- ($(c2)!0.2!(a1)$);
\draw [red] plot [smooth cycle] coordinates {($(b2)!0.4!(c1)$) ($(b2)!0.4!(c2)$) ($(b1)!0.4!(a2)$)
(-150:\big+3) (-180:\big+3)};
\draw [olive] plot [smooth cycle] coordinates {($(b2)!0.2!(c1)$) ($(b2)!0.2!(c2)$) ($(b2)!0.6!(b3)$) ($(b2)!0.4!(b1)$) (-180:\big+1.5)};
\draw [teal] plot [smooth cycle] coordinates {($(b1)!0.2!(a2)$) ($(b1)!1.2!(b3)$) ($(b3)!0.4!(b2)$) ($(b1)!0.4!(b2)$) (-150:\big+1.5)};
\draw [blue] plot [smooth cycle] coordinates {($(a2)!0.4!(b1)$) ($(a1)!0.4!(c2)$) (-0:\big+3) (-30:\big+3)};
\draw [violet] plot [smooth cycle] coordinates {($(a2)!0.2!(b1)$) ($(a2)!0.5!(a1)$) (-30:\big+1.5)};
\draw [orange] plot [smooth cycle] coordinates {($(a1)!0.2!(c2)$) ($(a1)!0.5!(a2)$) (-0:\big+1.5)};
\draw [black] plot [smooth cycle] coordinates {($(c2)!0.2!(a1)$) ($(c2)!0.2!(b2)$) ($(c1)!0.2!(b2)$)
(-255:\big+3) (-285:\big+3)};
\draw [cyan] plot [smooth cycle] coordinates {($(c2)!0.2!(a1)$) ($(c2)!0.2!(b2)$) ($(c2)!0.6!(c1)$) (-285:\big+1.5)};
\draw [purple] plot [smooth cycle] coordinates {($(c1)!0.2!(b2)$) ($(c1)!0.4!(c2)$) (-255:\big+1.5)};
\node[red] (x) at (-165:\big+5) {$S_1$};
\node[blue] (x) at (-15:\big+5) {$S_3$};
\node[black] (x) at (-270:\big+5) {$S_2$};
\end{tikzpicture}
\caption{Illustration of the \staticcoloring~and the coloring used for Steiner forest problem. 
In the graph, $S_2$ is inactive and does not color its cutting edges, while $S_1$ colors edges in red and $S_3$ colors in blue.
It is worth noting that edges within a connected component will not be further colored and will not be added to $\F$.}
\label{fig:growth}
\end{figure}
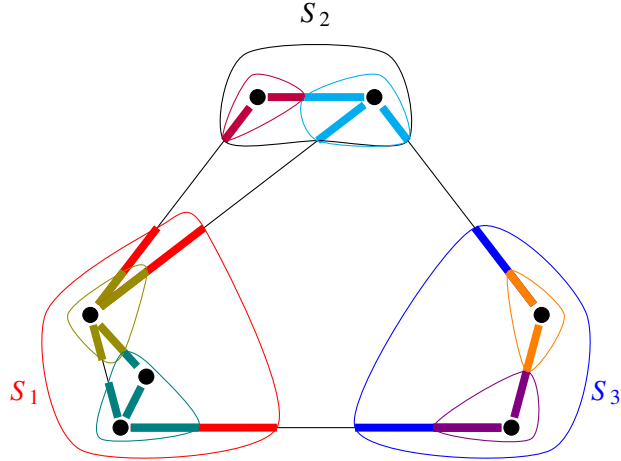 

At the end of the algorithm, we construct $\Fp$ from $\F$ by removing every edge that is not part of any path between the endpoints of any demand. 
We then analyze the cost of the optimal solution and our algorithm.
Let $\ys$ represent the amount of time that a set $S$ was active and colored its cutting edges. 
We can show that the cost of the optimal solution is at least $\sum_{S \subset \V} \ys$, while our algorithm will find a solution with a cost of at most $2 \sum_{S \subset \V} \ys$.

For each active set $S$, there exists at least one edge in the optimal solution that has exactly one endpoint inside $S$ and the other endpoint outside.
This is due to the fact that every active set cuts at least one demand, and the optimal solution must connect all demands.
While a set $S$ is active, it colors all of its cutting edges. As the optimal solution includes a cutting edge from $S$, we can conclude that an amount of $\ys$ from the optimal solution is colored by $S$. Since each set colors an uncolored portion of the edges, the cost of the optimal solution is at least $\sum_{S \subset \V} \ys$.

Furthermore, considering the fixed final forest $\Fp$, we can observe that at each moment of coloring, when we contract each connected component in $\currentsets$, it results in a forest where every leaf corresponds to an active set. 
This observation is based on the fact that if a leaf does not correspond to an active set, it implies that the only edge adjacent to that leaf is unnecessary and should have been removed from $\Fp$. 
Based on this insight, we can conclude that the number of edges being colored from $\Fp$ at that moment, which is equivalent to the sum of the degrees of the active sets in the aforementioned forest, is at most twice the number of active sets at that moment. 
This means that the amount of the newly colored portion of all edges at that moment is at most twice the total value added to all $\ys$. 
Therefore, considering that every edge in $\Fp$ is fully colored, we can deduce that the total length of edges in $\Fp$ is at most $2 \sum_{S \subset \V} \ys$.


\paragraph{A 3-approximation algorithm for prize collecting Steiner forest.}
Similar to the Steiner forest problem, we utilize coloring to solve PCSF. 
In PCSF, we encounter penalties that indicate it is not cost-effective to connect certain pairs $(i, j)$ if the cost exceeds a specified threshold $\pij$. 
To address this challenge, we introduce a coloring schema that assigns a color to each pair $(i, j)$, ensuring that the color is not used to color edges for a duration exceeding its associated potential $\pij$.
However, assigning colors to pairs introduces some challenges.
We are not aware of the distribution of potential between the endpoints of a pair, and each set may cut multiple pairs, making it unclear which color should be used at each moment.

To address these challenges, we use two types of coloring. The first type is called \staticcoloring, which is similar to the coloring schema used in the Steiner forest problem. In \staticcoloring, each set $S \subset \V$ is assigned a distinct color. It is referred to as \staticcoloring~since the colors assigned to edges corresponding to a set $S$ remain unchanged throughout the algorithm.
The second type is \dynamiccoloring, which involves coloring edges based on pairs $(i, j) \in \V \times \V$. Unlike \staticcoloring, \dynamiccoloring~allows the colors of edges to change during the algorithm, adapting to the evolving conditions.
By utilizing both \staticcoloring~and \dynamiccoloring, we can effectively handle the coloring requirements of PCSF, accounting for the potential constraints and varying edge coloring needs.

Similar to the Steiner forest algorithm, we begin by running the \staticcoloring~procedure. 
Whenever an edge is fully colored, we add it to the forest $\F$. 
However, unlike \staticcoloring, we do not maintain a separate \dynamiccoloring~throughout the algorithm, as it would require constant reconstructions. 
Instead, we compute the \dynamiccoloring~whenever needed.
To obtain the \dynamiccoloring, we map each moment of coloring for each set $S$ in the \staticcoloring~to a pair $(i, j)$ such that $S \odot (i, j)$, which means $S$ cuts $(i, j)$.
This assignment is achieved using a maximum flow algorithm, as described in Section \ref{sec:flow}. 
We ensure that our \staticcoloring~can always be converted to a \dynamiccoloring.
Let $\yij$ represent the total duration assigned to pair $(i, j)$ in the \dynamiccoloring. 
It is important to ensure that $\yij$ does not exceed the potential $\pij$ associated with that pair. 
If we encounter an active set $S$ for which assigning further coloring to any pair that $S$ cuts would exceed the pair's potential, we deactivate $S$ by removing it from $\activesets$.

We define a pair as ``tight'' if $\yij = \pij$. 
At the end of the algorithm, when every set is inactive, our goal is to pay the penalty for every tight pair. 
To minimize the number of tight pairs, we perform a local operation by assigning an $\epsilon$ amount of color assignment for set $S$ from pair $(i, j)$ to another pair $(i', j')$, such that $(i, j)$ was tight and after the operation, both pairs are no longer tight. 
Finally, we pay the penalty for every tight pair and construct $\Fp$ from $\F$ by removing any edges that are not part of a path between pairs that are not tight.
It is important to note that if a pair is not tight, it should be connected in $\F$. Otherwise, the sets containing the endpoints of that pair would still be active.
Thus, every pair is either connected or we pay its penalty.
Let us assume the optimal solution chooses forest $\fopt$ and pays penalties for pairs in $\qopt$.

Since we do not assign more color to each pair $(i, j)$ than its corresponding potential $\pij$, i.e., $\yij \le \pij$, we can conclude that the optimal solution pays at least $\sum_{(i, j) \in \qopt} \yij$ in penalties.
Moreover, similar to the argument for the Steiner forest, the cost of $\fopt$ is at least $\sum_{(i, j) \notin \qopt} \yij$.
Therefore, the cost of the optimal solution is at least $\sum_{S \subset \V} \ys = \sum_{(i, j) \in \V \times \V} \yij$, while, similar to the argument for the Steiner forest, the cost of $\Fp$ is at most $2\sum_{S \subset \V} \ys$.
Moreover, the total penalty we pay is at most $\sum_{(i, j) \in \V \times \V} \yij$, since we only pay for tight pairs. 
This guarantees a 3-approximation algorithm.


\paragraph{A 2-approximation algorithm for prize collecting Steiner forest.}
Let's refer to our 3-approximation algorithm as \pcsfthree. 
Our goal is to construct a 2-approximation algorithm called \rpcsf, by iteratively invoking \pcsfthree.
In \rpcsf, we first invoke \pcsfthree~and obtain a feasible solution $(\Q_1, \Fp_1)$, where $\Q_1$ represents the pairs for which we pay their penalty, and $\Fp_1$ is a forest that connects the remaining pairs. 
Next, we set the penalty for each pair in $\Q_1$ to $0$.
We recursively call \rpcsf~with the updated penalties. Let's assume that $(\Q_2, \Fp_2)$ is the result of this recursive call to \rpcsf~for the updated penalties. 
It is important to note that $(\Q_2, \Fp_2)$ is a feasible solution for the initial instance, as it either connects the endpoints of each pair or places them in $\Q_2$.
Furthermore, it is true that $\Q_1 \subseteq \Q_2$, as the penalty of pairs in $\Q_1$ is updated to $0$, and they will be considered as tight pairs in further iterations of $\pcsfthree$.
By induction, we assume that $(\Q_2, \Fp_2)$ is a 2-approximation of the optimal solution for the updated penalties.
Now, we want to show that either $(\Q_1, \Fp_1)$ or $(\Q_2, \Fp_2)$ is a 2-approximation of the optimal solution for the initial instance.
We will select the one with the lower cost and return it as the output of the algorithm.

To analyze the algorithm, we focus on the \dynamiccoloring~of pairs in $\Q_1$ that are connected in the optimal solution. 
Let $\B$ denote the set of pairs $(i, j) \in \Q_1$ that are connected in the optimal solution. 
We concentrate on this set because the optimal solution connects these pairs, and we will pay their penalties in both $(\Q_1, \Fp_1)$ and $(\Q_2, \Fp_2)$.
Let's assume $\vb$ represents the total duration that we color with a pair in $\B$ in \dynamiccoloring. 
Each moment of coloring with a pair $(i,j) \in \B$ in \dynamiccoloring~corresponds to coloring with a set $S$ in \staticcoloring~such that $S \odot (i,j)$. Since $(i,j) \in \B$ is connected in the optimal solution, we know that $S$ cuts at least one edge of the optimal solution, and $S$ colors that edge in \staticcoloring, while $(i,j)$ colors that edge in \dynamiccoloring. 
Thus, for any moment of coloring with pair $(i,j) \in \B$ in \dynamiccoloring, we will color at least one edge of the optimal solution.
Let $\vba$ be the total duration when pairs in $\B$ color exactly one edge of the optimal solution, and $\vbb$ be the total duration when pairs in $\B$ color at least two edges.
It follows that $\vba + \vbb = \vb$.

We now consider the values of $\vba$ and $\vbb$ to analyze the algorithm. 
If $\vbb$ is sufficiently large, we can establish a stronger lower bound for the optimal solution compared to our previous bound, which was $\sum_{S \subset V} \ys$. 
In the previous bound, we showed that each moment of coloring covers at least one edge of the optimal solution. 
However, in this case, we can demonstrate that a significant portion of the coloring process covers at least two edges at each moment. 
This improved lower bound allows us to conclude that the output of \pcsfthree, $(\Q_1, \Fp_1)$, becomes a 2-approximate solution.

Alternatively, if $\vba$ is significantly large, we can show that the optimal solution for the updated penalties is substantially smaller than the optimal solution for the initial instance. 
This is achieved by removing the edges from the initial optimal solution that are cut by sets whose color is assigned to pairs in $\B$ and that set only colored one edge of the optimal solution. 
By minimizing the number of tight pairs at the end of $\pcsfthree$, we ensure that no pair with a non-zero penalty is cut by any of these sets, and removing these edges will not disconnect those pairs.
Consequently, we can construct a feasible solution for the updated penalties without utilizing any edges from the cutting edges of these sets in the initial optimal solution. 
In summary, since $(\Q_2, \Fp_2)$ is a 2-approximation of the optimal solution for the updated penalties, and the optimal solution for the updated penalties has a significantly lower cost than the optimal solution for the initial instance, $(\Q_2, \Fp_2)$ becomes a 2-approximation of the optimal solution for the initial input.

Last but not least, we conduct further analysis of our algorithm to achieve a more refined approximation factor of $2-\frac{1}{n}$, which asymptotically approaches $2$.




\subsection{Preliminaries}

For a given set $S \subset \V$, we define the set of edges that have exactly one endpoint in $S$ as the \textit{cutting edges} of $S$, denoted by $\deltaS$. In other words, $\deltaS = \{(u, v) \in \E: |\{u, v\} \cap S| = 1\}$.
We say that $S$ cuts an edge $e$ if $e$ is a cutting edge of $S$, i.e., $e \in \deltaS$.
We say that $S$ cuts a forest $F$ if there exists an edge $e \in F$ such that $S$ cuts that edge.

For a given set $S \subset \V$ and pair $\{i, j\} \in \V \times \V$, we say that $S$ cuts $(i, j)$ if and only if $|\{i, j\} \cap S| = 1$. 
We denote this relationship as $S \odot (i, j)$.

For a forest $F$, we define $\cc(F)$ as the total cost of edges in $F$, i.e., $\cc(F) = \sum_{e \in F} \ce$. 

For a set of pairs of vertices $Q \subseteq \V \times \V$, we define $\pi(Q)$ as the sum of penalties of pairs in $Q$, i.e., $\pi(Q) = \sum_{(i, j) \in Q} \pij$.

For a given solution $\sol$ to a PCSF instance $\I$, the notation $cost(\sol)$ is used to represent the total cost of the solution.
In particular, if $\sol$ uses a forest $\F$ and pays the penalties for a set of pairs $\Q$, then the total cost is given by $cost(\sol)=\cc(\F)+\pi(\Q)$.

For a graph $G = (V, E)$ and a vertex $v \in \V$, we define $d_G(v)$ as the degree of $v$ in $G$. 
Similarly, for a set $S \subset V$, we define $d_G(S)$ as the number of edges that $S$ cuts, i.e., $|E \cap \deltaS|$.

Since we use max-flow algorithm in Section \ref{sec:flow}, we provide a formal definition of the \mincut~function:

\begin{definition}[\mincut]  
\label{def:max-flow}
    For the given directed graph $\G$ with source vertex $\source$ and sink vertex $\sink$, the function $\mincut(\G, \source, \sink)$ calculates the maximum flow from $\source$ to $\sink$ and returns three values: $\mfout$. 
    Here, $\mf$ represents the maximum flow value achieved from $\source$ to $\sink$ in $\G$, $\MC$ represents the min-cut between $\source$ and $\sink$ in $\G$ that minimizes the number of vertices on the $\source$ side of the cut, and $\f$ is a function $\f: E \rightarrow \mathbb{R}^{+}$ that assigns a non-negative flow value to each edge in the maximum flow.
\end{definition}

Throughout this paper, it is important to note that whenever we refer to the term ``minimum cut'' or ``min-cut,'' we specifically mean the minimum cut that separates $\source$ from $\sink$. 
Furthermore, we refer to the minimum cut that minimizes the number of vertices on the $\source$ side of the cut as the ``minimal min-cut''.


\section{Representing a 3-approximation Algorithm}
\label{sec:3_apx_alg}
In this section, we present an algorithm that utilizes coloring to obtain a 3-approximate solution.
Although the main part of this algorithm closely follows the approach presented by Hajiaghayi and Jain in \cite{DBLP:conf/soda/HajiaghayiJ06}, our novel interpretation of the algorithm is crucial for the subsequent analysis of the 2-approximation algorithm in the next section.
Furthermore, we introduce a modification at the end of the 3-approximation algorithm, which plays a vital role in achieving a 2-approximation algorithm in the next section.

From this point forward, we consider each edge as a curve with a length equal to its cost.
In our algorithm, we use two types of \coloring s: \staticcoloring~and \dynamiccoloring. Both of these \coloring s are used to assign colors to the edges of the graph, where each part of an edge is assigned a specific color. 
It is important to note that both \coloring s have the ability to assign different colors to different portions of the same edge.

First, we introduce some variables that are utilized in Algorithm \ref{alg:pcsf3}, and we will use them to define the \coloring s.
Let $\F$ be a forest that initially is empty, and we are going to add edges to in order to construct a forest that is a superset of our final forest.
Moreover, we maintain the set of connected components of $\F$ in $\currentsets$, where each element in $\currentsets$ represents a set of vertices that forms a connected component in $\F$.
Additionally, we maintain a set of active sets $\activesets \subseteq \currentsets$, which will be utilized for coloring edges in \staticcoloring. 
We will provide further explanation on this later.
Initially, we set $\activesets = \currentsets$.

\paragraph{Static Coloring.}
We construct an instance of \staticcoloring~iteratively by assigning colors to portions of edges. 
In \staticcoloring, each set $S \subset \V$ is assigned a unique color. 
Once a portion of an edge is colored in \staticcoloring, its color remains unchanged.

During the algorithm's execution, active sets color their cutting edges simultaneously and at the same speed, using their respective unique colors.
As a result, only edges between different connected components are colored at any given moment. 
When an edge $e$ is fully colored, we add it to $F$ and update $\currentsets$ to maintain the connected components of $F$. 
Since $e$ connects two distinct connected components of $F$, $F$ remains a forest. 
Furthermore, we update $\activesets$ by removing sets that contain an endpoint of $e$ and replacing them with their union.
Within the loop at Line \ref{line:merge_tight_edges_start} of Algorithm \ref{alg:pcsf3}, we check if an edge has been completely colored, and then merge the sets that contain its endpoints.
In addition, we provided a visual representation of the \staticcoloring~process in Figure \ref{fig:growth}.

\begin{definition}[Static coloring duration]
    For an instance of \staticcoloring, define $\ys$ as the duration during which set $S$ colors its cutting edges using the color $S$.
\end{definition}

It is important to note that we do not need to store the explicit portion of each edge that is colored. 
Instead, we keep track of $\ys$, which represents the amount of coloring associated with set $S$. The portion of edge $e$ that is colored can be computed as $\sum_{S:e\in \deltaS} \ys$.

Now, we will explain the procedure $\finddeltae$, which determines the first moment in time, starting from the current moment, when at least one new edge will become fully colored. 
This procedure is essential for executing the algorithm in discrete steps.

\paragraph{Finding the maximum value for \Deltae.}
In \finddeltae, we determine the maximum value of $\Deltae$ such that continuing the coloring process for an additional duration $\Deltae$ does not exceed the length of any edges.
We consider each edge $e = (v, u)$ where $v$ and $u$ are not in the same connected component, and at least one of them belongs to an active set. 
The portion of edge $e$ that has already been colored is denoted by $\sum_{S: e \in \deltaS} \ys$. 
The remaining portion of edge $e$ requires a total time of $(\ce - \sum_{S: e \in \deltaS} \ys)/t$ to be fully colored, where $t$ is the number of endpoints of $e$ that are in an active set.
It is important to note that the coloring speed is doubled when both endpoints of $e$ are in active sets compared to the case where only one endpoint is in an active set.
To ensure that the edge lengths are not exceeded, we select $\Deltae$ as the minimum time required to fully color an edge among all the edges.

\begin{corollary}
\label{col:deltae-fully-colored}
    After coloring for $\Deltae$ duration, at least one new edge becomes fully colored.
\end{corollary}

In Algorithm \ref{alg:finddeltae}, we outline the procedure for $\finddeltae$.

\begin{algorithm}[ht]
  \caption{Fidning the maximum value for \Deltae}
  \label{alg:finddeltae}
  \hspace*{\algorithmicindent} \textbf{Input:} An undirected graph $\G=(\V, \E, \cc)$ with edge costs $\cc: \E \rightarrow \mathbb{R}_{\ge 0}$, an instance of \staticcoloring\\\hspace*{\algorithmicindent}\hspace{12mm} represented by $\y : 2^{\V} \rightarrow \mathbb{R}_{\ge 0}$, active sets \activesets, and connected components \currentsets. \\
  \hspace*{\algorithmicindent} \textbf{Output:} $\Deltae$, the maximum value that can be added to $\ys$ for $S \in \activesets$ without violating edge lengths.
  \begin{algorithmic}[1]
    \Procedure{\finddeltae}{\G, \y, \activesets, \currentsets}
        \State Initialize $\Deltae \gets \infty$
        \For{$e\in E$} 
        \label{line:delta_e_loop}
            \State Let $S_v, S_u \in \currentsets$ be the sets that contain each endpoint of $e$.
            \State $t \gets |\{S_v, S_u\} \cap \activesets|$
            \If{$S_v \neq S_u$ \textbf{and}~$t \neq 0$}
                \State $\Deltae \gets \min (\Deltae, (\ce- \sum_{S: e \in \deltaS} \ys)/t)$ \label{line:assign_delta_e}
            \EndIf
        \EndFor
        \State \Return $\Deltae$
    \EndProcedure
  \end{algorithmic}
\end{algorithm}

Now, we can utilize $\finddeltae$ to perform the coloring in discrete steps, as shown in Algorithm \ref{alg:pcsf3}. 
In summary, during each step, at Line \ref{line:call_find_delta_e}, we call $\finddeltae$ to determine the maximum duration $\Deltae$ for which we can color with active sets without exceeding the length of any edge, ensuring that at least one edge will be fully colored.
Similarly, at Line \ref{line:assign_delta_p}, we utilize $\finddeltap$ to determine the maximum value of $\Deltap$ that ensures a valid \staticcoloring~when extending the coloring duration by $\Deltap$ using active sets. 
The concept of a valid \staticcoloring, which avoids purchasing edges when it is more efficient to pay penalties, will be further explained in Section \ref{sec:flow}.

Then, at Line \ref{line:grow_active_sets}, we advance the static coloring process for a duration of $\min(\Deltae, \Deltap)$. 
In the subsequent loop at Line \ref{line:merge_tight_edges_start}, we identify newly fully colored edges and merge their endpoints' sets.
Additionally, within the loop at Line \ref{line:loop_deactive}, we will identify and deactivate sets that should not remain active, as their presence would lead to an invalid \staticcoloring.
We continue updating our \staticcoloring~until no active sets remain. 
Finally, we set $\Q$ equal to the set of pairs for which we need to pay penalties in Line \ref{line:call_reduce_tight_pairs}, and we derive our final forest $\Fp$ from $\F$ by removing redundant edges that are not necessary for connecting demands in $(\V \times \V)\setminus \Q$.

In Algorithm \ref{alg:pcsf3}, we utilize three functions other than $\finddeltae$: $\finddeltap$, $\checksetistight$, and $\reducetightpairs$.
The purpose of $\finddeltap$ is to determine the maximum value of $\Deltap$ that allows for an additional coloring duration of $\Deltap$ resulting in a \validstaticcoloring. 
$\checksetistight$ is responsible for identifying sets that cannot color their cutting edges while maintaining the validity of the static coloring.
Lastly, $\reducetightpairs$ aims to reduce the number of pairs for which penalties need to be paid and determine the final set of pairs that we pay their penalty. 
All of these functions utilize \dynamiccoloring, which will be explained in Section \ref{sec:flow}. 

It is important to note that we do not store a \dynamiccoloring~within \pcsfthree~since it changes constantly. 
Instead, we compute a \dynamiccoloring~based on the current \staticcoloring~within these functions, as they are the only parts of our algorithm that require a \dynamiccoloring.
Note that at the end of $\pcsfthree$, we require a final \dynamiccoloring~for the analysis in Section \ref{sec:2_apx_alg}. 
This final coloring will be computed in $\reducetightpairs$ at Line \ref{line:call_reduce_tight_pairs}.

\begin{algorithm}[ht]
  \caption{A 3-approximation Algorithm}
  \label{alg:pcsf3}
  \hspace*{\algorithmicindent} \textbf{Input:} An undirected graph $\G=(\V, \E, \cc)$ with edge costs $\cc: \E \rightarrow \mathbb{R}_{\ge 0}$ and penalties $\pi : \V \times \V \rightarrow \mathbb{R}_{\ge 0}$. \\
  \hspace*{\algorithmicindent} \textbf{Output:} A set of pairs $\Q$ with a forest $\Fp$ that connects the endpoints of every pair $(i, j) \notin \Q$.
  \begin{algorithmic}[1]
    \Procedure{\pcsfthree}{$I=(\G, \pi)$}
      \State Initialize $\F \gets \emptyset$
      \State Initialize $\activesets, \currentsets \gets \{\{v\}: v \in \V\}$
      \label{line:init-acts}
      \State Implicitly set $\ys \gets \emptyset$ for all $S \subset \V$
      \While{$\activesets \neq \emptyset$} \label{line:pcsf3-while}
        \State $\Deltae \gets \finddeltae(\G, \y, \activesets, \currentsets) $ \label{line:call_find_delta_e}
        \State $\Deltap \gets \finddeltap(\G, \pi, \y, \activesets) $ \label{line:assign_delta_p}
        \State $\Deltat \gets \min(\Deltae, \Deltap)$
        \For{$S \in \activesets$}
            \State $\ys \gets \ys + \Deltat$ \label{line:grow_active_sets}
        \EndFor
        \For{$e\in E$} \label{line:merge_tight_edges_start}
          \State Let $S_v, S_u \in \currentsets$ be sets that contains each endpoint of $e$
          \If{$\sum_{S: e \in \deltaS} \ys = \ce$ \textbf{and} $S_v \neq S_u$} 
            \State $\F \gets \F \cup \{e\}$
            \State $\currentsets \gets (\currentsets \setminus \{S_p, S_q\}) \cup \{S_p \cup S_q\}$
            \label{line:update-c}
            \State $\activesets \gets (\activesets \setminus \{S_p, S_q\}) \cup \{S_p \cup S_q\}$ \label{line:merge_tight_edges_end}
            \label{line:update-as}
          \EndIf
        \EndFor
        \For{$S \in \activesets$} \label{line:loop_deactive}
          \If{$\checksetistight(\G, \pi, \y, S)$} \label{line:find_tight_sets}
            \State $\activesets \gets \activesets \setminus \{S\}$ \label{line:deactivate_tight_sets}
          \EndIf
        \EndFor
      \EndWhile
      \State $\Q \gets \reducetightpairs(\G, \pi, \y)$ \label{line:call_reduce_tight_pairs}
      \State Let $\Fp$ be the subset of $\F$ obtained by removing unnecessary edges for connecting demands $(\V \times \V) \setminus \Q$.\label{line:create_F'}
      \State \Return $(\Q, \Fp)$
    \EndProcedure
  \end{algorithmic}
\end{algorithm}

Now, let's analyze the time complexity of $\finddeltae$ as described in Lemma \ref{lm:finddeltae-polynomial}.
In Lemma \ref{lm:pcsfthree_polynomial}, we will demonstrate that the overall time complexity of $\pcsfthree$ is polynomial. This will be achieved after explaining and analyzing the complexity of the subroutines it invokes.

\begin{lemma}
\label{lemma:linear-active}
    In the \pcsfthree~algorithm, the number of sets that have been active at some point during its execution is linear.
\end{lemma}
\begin{proof}
    During the algorithm, new active sets are only created in Line \ref{line:update-as} by merging existing sets. 
    Initially, we start with $n$ active sets in \activesets. 
    Symmetrically, for each creation of a new active set, we have one merge operation over sets in \currentsets, which reduces the number of sets in $\currentsets$ by exactly one. 
    Since we start with $n$ sets in \currentsets, the maximum number of merge operations is $n-1$. Therefore, the total number of active sets throughout the algorithm is at most $2n-1$.
\end{proof}

\begin{lemma}
\label{lm:finddeltae-polynomial}
    The runtime of \finddeltae~is polynomial.
\end{lemma}
\begin{proof}
In Line \ref{line:delta_e_loop}, we iterate over the edges, and the number of edges is polynomial. 
In addition, since each set $S$ with $\ys > 0$ has been active at some point, the number of these sets is linear due to Lemma \ref{lemma:linear-active}. Consequently, for each edge, we calculate the sum in Line \ref{line:assign_delta_e} in linear time by iterating through such sets. Therefore, we can conclude that \finddeltae~runs in polynomial time.
\end{proof}


\subsection{Dynamic Coloring}
\label{sec:flow}

The \dynamiccoloring~is derived from a given \staticcoloring.
In \dynamiccoloring, each pair $(i, j) \in \V \times \V$ is assigned a unique color.
The goal is to assign each moment of coloring in \staticcoloring~with each active set $S \subset \activesets$, to a pair $(i, j) \in \V \times \V$ where $S \odot (i, j)$ holds, and color the same portion that set $S$ colored at that moment in \staticcoloring~with the color of pair $(i, j)$ in \dynamiccoloring.
Furthermore, there is a constraint on the usage of each pair's color. 
We aim to avoid using the color of pair $(i, j)$ for more than a total duration of $\pij$. It's important to note that for a specific \staticcoloring, there may be an infinite number of different \dynamiccoloring s, but we only need to find one of them.

Now, we will introduce some notations that are useful in our algorithm and analysis.

\begin{definition}[Dynamic Coloring Assignment Duration]
    In a \dynamiccoloring~instance, for each set $S$ and pair $(i, j)$ where $S \odot (i, j)$, $\ysij$ represents the duration of coloring with color $S$ in \staticcoloring~that is assigned to pair $(i, j)$ for coloring in \dynamiccoloring.
\end{definition}
    
\begin{definition}[Dynamic Coloring Duration]
    In a \dynamiccoloring~instance, $\yij$ represents the total duration of coloring with pair $(i, j)$ in \dynamiccoloring, denoted as $\yij = \sum_{S: S \odot (i, j)} \ysij$.
\end{definition}

\begin{definition}[Pair Constraint and Tightness]
\label{def:tightness}
    In a \dynamiccoloring~instance, the condition that each pair $(i, j)$ should not color for more than $\pij$ total duration (i.e., $\yij \le \pij$) is referred to as the \pairconst.
    If this condition is tight in the \dynamiccoloring~for a pair $(i, j)$, i.e., $\yij = \pij$, we say that pair $(i, j)$ is a tight pair.
\end{definition}

\begin{definition}[Valid Static Coloring]
\label{def:valid-static-coloring}
    A \staticcoloring~is considered valid if there exists a \dynamiccoloring~for the given \staticcoloring. In other words, the following conditions must hold:
    \begin{itemize}
        \item For every set $S \subset V$, we can distribute the duration of the static coloring for $S$ among pairs $(i, j)$ that satisfy $S \odot (i, j)$, such that $\sum_{(i, j): S \odot (i, j)} \ysij = \ys$.
        \item For every pair $(i, j) \in \V \times \V$, the pair constraint is not violated, i.e., $\yij = \sum_{S: S\odot (i, j)} \ysij \le \pij$.
    \end{itemize}
    If there is no \dynamiccoloring~that satisfies these conditions, the \staticcoloring~is considered invalid.
\end{definition}

Note that in the definition of \validstaticcoloring, the validity of a \staticcoloring~is solely determined by the duration of using each color, denoted as $\ys$, and the specific timing of using each color is not relevant.
Moreover, a function $\y : 2^{\V} \rightarrow \mathbb{R}_{\ge 0}$ is almost sufficient to describe a \staticcoloring, as it indicates the duration for which each set $S$ colors its cutting edges.
Therefore, this function provides information about the portion of each edge that is colored with each color.
This information is enough for our algorithm and analysis. 
We are not concerned with the precise location on an edge where a specific color is applied.
Instead, our focus is on determining the amount of coloring applied to each edge with each color.
Similarly, the function $\y : 2^{\V} \times \V \times \V \rightarrow \mathbb{R}_{\ge 0}$ is enough for determining a \dynamiccoloring.

\begin{definition}[Set Tightness]
\label{def:set-tightness}
    For a valid instance of \staticcoloring, we define a set $S \subset \V$ as tight if increasing the value of $\ys$ by any $\epsilon>0$ in the \staticcoloring~without changing the coloring duration of other sets would make the \staticcoloring~invalid.
\end{definition}

\begin{lemma}
\label{lemma:tight-set-tight-pairs}
    In a \validstaticcoloring, if set $S$ is tight, then for any corresponding \dynamiccoloring, all pairs $(i, j)$ such that $S \odot (i, j)$ are tight.  
\end{lemma}
\begin{proof}
    Consider an arbitrary \dynamiccoloring~of the given \validstaticcoloring.
    Using contradiction, assume there is a pair $(i, j)$ such that $S \odot (i, j)$, and this pair is not tight.
    Let $\epsilon = \pij - \yij$. 
    If we increase $\ys$, $\ysij$, and $\yij$ by $\epsilon$, it results in a new \staticcoloring~and a new \dynamiccoloring. 
    In the new \dynamiccoloring, since for every set $S'$ we have $\sum_{(i', j'): S' \odot (i', j')} \y_{S'i'j'} = \y_{S'}$, and for every pair $(i', j')$ we have $\y_{i'j'} = \sum_{S': S' \odot (i', j')} \y_{S'i'j'} \le \pi_{i'j'}$, based on Definition \ref{def:valid-static-coloring}, increasing $\ys$ by $\epsilon$ results in a valid \staticcoloring.
    According to Definition \ref{def:set-tightness}, this contradicts the tightness of $S$. 
    Therefore, we can conclude that all pairs $(i, j)$ for which $S \odot (i, j)$ holds are tight.
\end{proof}

However, it is possible for every pair $(i, j)$ satisfying $S \odot (i, j)$ to be tight in a \dynamiccoloring, while the set $S$ itself is not tight.
We will describe the process of determining whether a set is tight in Algorithm \ref{alg:check_set_tight}.

So far, we have explained several key properties and concepts related to \dynamiccoloring.
Now, let us explore how we can obtain a \dynamiccoloring~from a given \validstaticcoloring.
To accomplish this, we introduce the concept of $\setpairgraph$, which represents a graph associated with each \staticcoloring.
By applying the max-flow algorithm to this graph, we can determine a corresponding \dynamiccoloring.
The definition of $\setpairgraph$ is provided in Definition \ref{def:setpairgraph}, and Figure \ref{fig:set-pair-graph} illustrates this graph.

\begin{definition}[\setpairgraph]
\label{def:setpairgraph}
    Given a graph $\G=(\V, \E, \cc)$ with edge costs $\cc: \E \rightarrow \mathbb{R}_{\ge 0}$, and penalties $\pi : \V \times \V \rightarrow \mathbb{R}_{\ge 0}$, as well as an instance of \staticcoloring~represented by $\y : 2^{\V} \rightarrow \mathbb{R}_{\ge 0}$, we define a directed graph $\GF$ initially consisting of two vertices $\source$ and $\sink$.
    For every set $S \subset V$ such that either $\ys > 0$ or $S \in \activesets$, we add a vertex to $\GF$ and a directed edge from $\source$ to $S$ with a capacity of $\ys$.
    Additionally, for each pair $(i, j) \in \V \times \V$, we add a vertex to $\GF$ and a directed edge from $(i, j)$ to $\sink$ with a capacity of $\pij$.
    Finally, we add a directed edge from each set $S$ to each pair $(i, j)$ such that $S \odot (i, j)$, with infinite capacity.
    We refer to the graph $\GF$ as $\setpairgraph(\G, \pi, \y)$.
\end{definition}

\begin{figure}
	\centering
\begin{tikzpicture}[scale=0.15]
\def\dx{22}
\def\dy{4}
\def\r{0.7}
\foreach \nd/\x/\y/\tx in {0/0/0/source, 1/1/4/, 2/1/2/, 3/1/0/S, 4/1/-2/, 5/1/-4/, 6/2/3/, 7/2/1/, 8/2/-1/(i\comma j), 9/2/-3/, 10/3/0/sink} {
    \draw[fill=black] (\x*\dx, \y*\dy) circle (\r);
    \node[above] at (\x*\dx, \y*\dy+1) {$\tx$};
    \node (\nd) at (\x*\dx, \y*\dy) {};
};
\draw[->] (0) -- (3) node[midway, below,sloped] {$y_{S}$};
\draw[->] (3) -- (8) node[midway, below,sloped] {$\infty$} node[midway,above,sloped] {[if $S \odot (i, j)$]};
\draw[->] (8) -- (10) node[midway, below,sloped] {$\pij$};
\end{tikzpicture}
\caption{\setpairgraph}
\label{fig:set-pair-graph}
\end{figure}

Now, we compute the maximum flow from $\source$ to $\sink$ in $\setpairgraph(\G, \pi, \y)$.
We assign the value of $\ysij$ to the amount of flow from $S$ to $(i, j)$, representing the allocation of coloring from set $S$ in \staticcoloring~to pair $(i, j)$ in \dynamiccoloring. 
Similarly, we set $\yij$ equal to the amount of flow from $(i, j)$ to $\sink$, indicating the duration of coloring for pair $(i, j)$ in \dynamiccoloring.
In Lemma \ref{lm:valid_dynamic_coloring}, we show that if the maximum flow equals $\sum_{S \subset \V} \ys$, the \staticcoloring~is valid, and the assignment of $\ysij$ and $\yij$ satisfies all requirements.

\begin{lemma}
\label{lm:valid_dynamic_coloring}
    For a given graph $\G=(\V, \E, \cc)$, penalties $\pi : \V \times \V \rightarrow \mathbb{R}_{\ge 0}$, and an instance of \staticcoloring~represented by  $\y : 2^{\V} \rightarrow \mathbb{R}_{\ge 0}$, let $(\mf, \MC, f) = \mincut(\setpairgraph(\G, \pi, \y), \source, \sink)$,
    the \staticcoloring~is valid if and only if $\mf = \sum_{S \subset \V} \ys$.
\end{lemma}
\begin{proof}
    Assume that $\mf = \sum_{S \subset \V} \ys$. Since the sum of the capacities of the outgoing edges from $\source$ is equal to the maximum flow, the amount of flow passing through $S$ is $\ys$.
    Hence, given that the total flow coming out from $S$ is equal to $\sum_{(i, j) \in \V \times \V} \ysij$, we have $\ys = \sum_{(i, j) \in \V \times \V} \ysij$.
    It should be noted that $\ysij > 0$ only if $S \odot (i, j)$, as we only have a directed edge from $S$ to $(i, j)$ in that case.

    Furthermore, since the capacity of the edge from $(i, j)$ to $\sink$ is $\pij$, the sum of incoming flow to $(i, j)$ is at most $\pij$.
    Thus, $\yij = \sum_{S \subset \V: S \odot (i,j)} \ysij \le \pij$.

    Now, assume that the given \staticcoloring~is valid. 
    Consider its corresponding $\setpairgraph$ and \dynamiccoloring.
    Let the amount of flow from $\source$ to $S$ be $\ys$.
    Let the amount of flow in the edge between $S$ and $(i, j)$ be $\ysij$, representing the assignment duration in the \dynamiccoloring.
    Similarly, let the amount of flow in the edge between $(i, j)$ and $\sink$ be $\yij$, representing the duration in the \dynamiccoloring.
    According to Definition \ref{def:valid-static-coloring}, in a \validstaticcoloring, the following conditions hold: $\sum_{(i, j): S \odot (i, j)} \ysij = \ys$ and $\yij = \sum_{S: S\odot (i, j)} \ysij \le \pij$. 
    Therefore, in the \setpairgraph, the assignment of flow satisfies the edge capacities and the equality of incoming and outgoing flows for every vertex except $\source$ and $\sink$.
    Furthermore, since we fulfill all outgoing edges from $\source$, the maximum flow is $\sum_{S \subset \V} \ys$.
\end{proof}

Let us define $\SC$ as the cut in $\setpairgraph$ that separates $\source$ from other vertices.
Since the sum of the edges in $\SC$ is $\sum_{S \subset \V} \ys$, the following corollary can easily be concluded from Lemma \ref{lm:valid_dynamic_coloring}.

\begin{corollary}
\label{col:valid_dynamic_coloring}
    For a given graph $\G=(\V, \E, \cc)$, penalties $\pi : \V \times \V \rightarrow \mathbb{R}_{\ge 0}$, and an instance of \staticcoloring~represented by  $\y : 2^{\V} \rightarrow \mathbb{R}_{\ge 0}$,
    the \staticcoloring~is valid if and only if $\SC$ is a minimum cut between $\source$ and $\sink$ in $\setpairgraph(\G, \pi, \y)$.
\end{corollary}

To analyze the size of the \setpairgraph~and the complexity of running the max-flow algorithm on it, we can refer to the following lemma.

\begin{lemma}
\label{lm:size_of_setpairgraph}
    At any point during the algorithm, the size of the \setpairgraph~remains polynomial.
\end{lemma}
\begin{proof}
In the \setpairgraph, vertices are assigned to sets that are either active or have $\ys > 0$. This implies that for each set that is active at least once, there is at most one corresponding vertex in the graph. According to Lemma \ref{lemma:linear-active}, the number of such vertices is linear.

In addition to the active set vertices, the \setpairgraph~also includes vertices $\source$, $\sink$, and pairs $(i, j)$. The total number of such vertices is $2 + n^2$. Thus, the overall size of the graph is polynomial.
\end{proof}

Now that we understand how to find a \dynamiccoloring~given a \staticcoloring~using the max-flow algorithm, we can use this approach to develop the functions \finddeltap, \checksetistight, and \reducetightpairs.

\paragraph{Finding the maximum value for $\Deltap$.}
In $\finddeltap$, our goal is to determine the maximum value of $\Deltap$ such that if we continue coloring with active sets for an additional duration of $\Deltap$ in the \staticcoloring, it remains a \validstaticcoloring.
The intuition behind this algorithm is to start with an initial upper bound for $\Deltap$ and iteratively refine it until we obtain a \validstaticcoloring. 
This process involves adjusting the parameters and conditions of the coloring to gradually tighten the upper bound.
Algorithm \ref{alg:finding_deltap} presents the pseudocode for $\finddeltap$.
The proof of the following lemma illustrates how the iterations of this algorithm progress toward the correct value of $\Deltap$.

\begin{lemma}
\label{lemma:finddeltap-max}
In a \validstaticcoloring, the maximum possible duration to continue the coloring process while ensuring the validity of the \staticcoloring~is $\Deltap = \finddeltap(\G, \pi, \y, \activesets)$.
\end{lemma}
\begin{proof}
Consider the $\setpairgraph$ of the given \validstaticcoloring. If increasing the capacity of edges from $\source$ to $S$ for every set $S \in \activesets$ by $\Deltap$ results in an increase in the min-cut by $|\activesets| \cdot \Deltap$, then the resulting \staticcoloring~remains valid since $\SC$ remains a min-cut. This is based on Corollary \ref{col:valid_dynamic_coloring}.
Thus, our goal is to find the maximum value for $\Deltap$ that satisfies this condition.

First, in Line \ref{line:set_upperbound_deltap}, we initialize $\Deltap$ with an upper-bound value of $(\sum_{ij}\pij - \sum_{S}\ys)/|\activesets|$.
This value serves as an upper-bound because the increase in min-cut cannot exceed $( \sum_{ij}\pij - \sum_{S}\ys)$, as determined by the cut that separates $\sink$ from the other vertices.
Next, we update the capacity of edges from $\source$ to the active sets in $\setpairgraph$ by the value of $\Deltap$, and then calculate the minimal min-cut $\MC$.
If $\MC$ separates $\source$ from the other vertices, it indicates that the new \staticcoloring~is valid according to Corollary \ref{col:valid_dynamic_coloring}. 
At this point, the function terminates and returns the current value of $\Deltap$ in Line \ref{line:return_deltap}.

Otherwise, if $\SC$ is not a min-cut after updating the edges, we want to prove that $\MC$ has at least one active set in the side of $\source$.
Let's assume, for contradiction, that $\MC$ does not have any active sets on the side of $\source$. 
According to Corollary \ref{col:valid_dynamic_coloring}, prior to updating the edges, $\SC$ was a min-cut since we had a \validstaticcoloring. 
Thus, the weight of $\MC$ was at least equivalent to the weight of $\SC$ before the edges were modified.
Given that $\MC$ and $\SC$ include all the edges connecting $\source$ to the active sets, adding $\Deltap$ to active sets leads to an increase in the weight of both $\MC$ and $\SC$ by $|\activesets| \cdot \Deltap$.
Consequently, the weight of $\MC$ remains at least as large as the weight of $\SC$. 
However, if $\SC$ is not a minimum cut after updating the edges, it implies that $\MC$ cannot be a minimum cut either, which contradicts the definition of $\MC$.
Therefore, we can conclude that $\MC$ must have at least one active set on the same side as $\source$.

Let $k \ge 1$ represent the number of active sets on the same side as $\source$ in $\MC$.
Referring to Definition~\ref{def:max-flow}, it is evident that $\MC$ minimizes $k$. 
If we decrease $\Deltap$ by $\epsilon$, it deducts $|\activesets| \cdot \epsilon$ from the weight of $\SC$ and $(|\activesets| - k) \cdot \epsilon$ from weight of $\MC$.
Given that the weight of $\SC$ is $|\activesets| \cdot \Deltap + \sum_{S \subset V} \ys$, and the weight of $\MC$ is $\mf$, in order to establish $\SC$ as a minimum cut, we want to find the minimum value of $\epsilon$ satisfying:
\begin{align*}
    \mf - (|\activesets| - k) \cdot \epsilon 
    &\ge 
    |\activesets| \cdot \Deltap + \sum_{S \subset V} \ys - |\activesets| \cdot \epsilon 
    \\
    \mf + k \epsilon 
    &\ge 
    |\activesets| \cdot \Deltap + \sum_{S \subset V} \ys
    \\
    k\epsilon 
    &\ge
    |\activesets| \cdot \Deltap + \sum_{S \subset V} \ys - \mf
    \\ 
    \epsilon 
    &\ge
    \frac{|\activesets| \cdot \Deltap + \sum_{S \subset V} \ys - \mf}{k}
\end{align*}

Now, let us define $\epsilon^* = (|\activesets| \cdot \Deltap + \sum_{S \subset V} \ys - \mf)/k$. 
When we decrease $\Deltap$ by $\epsilon^*$, we effectively tighten the upper bound on $\Deltap$. 
This is crucial because reducing $\Deltap$ by a smaller value would make it impossible for $\SC$ to become a min-cut. 
Such a violation would contradict the validity of the \staticcoloring, as indicated by Corollary \ref{col:valid_dynamic_coloring}.
After updating $\Deltap$ to its new value, if $\SC$ indeed becomes a minimum cut, the procedure is finished. However, if the new minimal min-cut still contains active sets on the side of $\source$, their number must be less than $k$.

To prove this by contradiction, let's assume that in the new minimal min-cut, the number of active sets on the side of $\source$ is greater than or equal to $k$.
Due to the minimality of the new minimum cut, it can be observed that all minimum cuts for the updated $\Deltap$ have at least $k$ active sets on the $\source$ side. 
In other words, these minimum cuts have at most $|\activesets| - k$ active sets on the other side. 
Consequently, the weight of these minimum cuts is reduced by at most $(|\activesets| - k) \cdot \epsilon^*$.
Since we have specifically reduced $(|\activesets| - k) \cdot \epsilon^*$ from $\MC$, it remains a minimum cut. 
Moreover, after decreasing $\epsilon^*$ from $\Deltap$ based on how we determine $\epsilon^*$, $\MC$ and $\SC$ have the same weight. 
This implies that $\SC$ is also a valid minimum cut and should be the minimal min-cut.
This contradiction suggests that after the reduction, the number of vertices on the $\source$ side has indeed decreased.

Finally, we repeat the same procedure until $\SC$ becomes a min-cut. Given that there are at most $n$ active sets in \activesets, and each iteration reduces the number of active sets on the side of $\source$ in the minimal min-cut by at least one, after a linear number of iterations, all active sets will be moved to the other side, and the desired value of $\Deltap$ will be determined.
Since each time we have demonstrated that the value of $\Deltap$ serves as an upper bound, it represents the maximum possible value that allows for a \validstaticcoloring.
\end{proof}

\begin{lemma}
\label{lemma:finddeltap}
    In a \validstaticcoloring, the \staticcoloring~remains valid if we continue the coloring process by at most $\Deltap = \finddeltap(\G, \pi, \y, \activesets)$.
\end{lemma}
\begin{proof}
Consider the $\setpairgraph$ of the given \validstaticcoloring.
We want to show that for every value $\Deltap' \le \Deltap$, the coloring remains valid.
According to Lemma \ref{lemma:finddeltap-max}, we know that increasing the duration of the active sets by $\Deltap$ results in a \validstaticcoloring. 
Therefore, by increasing the capacity of edges from $\source$ to $S$ for each set $S \in \activesets$ by $\Deltap$, $\SC$ represents a minimum cut. Decreasing the capacities of these edges by a non-negative value $d = \Deltap - \Deltap'$ results in a decrease in the weight of $\SC$ by $|\activesets|\cdot d$, while other cuts are decreased by at most this value.
Consequently, $\SC$ remains a minimum cut, and as indicated in Corollary \ref{col:valid_dynamic_coloring}, the static coloring remains valid after increasing $\ys$ for active sets by $\Deltap' \leq \Deltap$.
\end{proof}

Based on the proof of Lemma \ref{lemma:finddeltap-max}, it is clear that the while loop in the \finddeltap~function iterates a linear number of times. Furthermore, during each iteration, we make a single call to the \mincut~procedure on the \setpairgraph, which has a polynomial size according to Lemma \ref{lm:size_of_setpairgraph}. Consequently, we can deduce the following corollary.

\begin{corollary}
\label{col:finddeltap-polynomial}
    Throughout the algorithm, each call to the \finddeltap~function executes in polynomial time.
\end{corollary}

Now, we can prove a significant lemma that demonstrates that our algorithm behaves as expected.

\begin{lemma}
\label{lemma:always-valid}
    Throughout the algorithm, we always maintain a \validstaticcoloring.
\end{lemma}
\begin{proof}
    We can establish the validity of the \staticcoloring~throughout the algorithm using induction. Initially, since $\ys = 0$ for all sets $S$, assigning $\ysij = 0$ and $\yij = 0$ results in a \dynamiccoloring, satisfying the conditions of a \validstaticcoloring.

    Assuming that at the beginning of each iteration, we have a \validstaticcoloring~based on the induction hypothesis. 
    Furthermore, during each iteration, we continue the coloring process for a duration of $\min(\Deltap, \Deltae) \le \Deltap$. 
    According to Lemma \ref{lemma:finddeltap}, this coloring preserves the validity of the \staticcoloring.

    Therefore, by induction, we can conclude that throughout the algorithm, we maintain a \validstaticcoloring.
\end{proof}

\begin{algorithm}[ht]
  \caption{Finding the maximum value for $\Deltap$}
  \label{alg:finding_deltap}
  \hspace*{\algorithmicindent} \textbf{Input:} An undirected graph $\G=(\V, \E, \cc)$ with edge costs $\cc: \E \rightarrow \mathbb{R}_{\ge 0}$, penalties $\pi : \V \times \V \rightarrow \mathbb{R}_{\ge 0}$, \\\hspace*{\algorithmicindent}\hspace{12mm} an instance of \staticcoloring~represented by $\y : 2^{\V} \rightarrow \mathbb{R}_{\ge 0}$, and active sets \activesets.  \\
  \hspace*{\algorithmicindent} \textbf{Output:} 
    $\Deltap$, the maximum value that can be added to $\ys$ for $S \in \activesets$ without violating the validity\\\hspace*{\algorithmicindent}\hspace{14.5mm} of the \staticcoloring.
  \begin{algorithmic}[1]
    \Function{\finddeltap}{$\G,\pi, \y, \activesets$}
    \State $\GF \gets \setpairgraph(\G, \pi, \y)$
    \State $\Deltap \gets \left( \sum_{ij}\pij - \sum_{S}\ys \right)/|\activesets|$ \label{line:set_upperbound_deltap}
    \While{true}
        \State Set the capacity of edges from $\source$ to $S \in \activesets$ equals to $\ys + \Deltap$ in graph $\GF$.
        \State $\mfout \gets \mincut(\GF, \source, \sink)$
        \If{$\mf = |\activesets| \cdot \Deltap + \sum_{S \subset V} \ys$}
            \State \Return $\Deltap$ \label{line:return_deltap}
        \EndIf
        \State Let $k$ represent the number of active sets on the same side of the cut $\MC$ as $\source$.
        \State $\Deltap \gets \Deltap - \left(|\activesets| \cdot \Deltap + \sum_{S \subset V} \ys - \mf \right)/k$
    \EndWhile
    \EndFunction
  \end{algorithmic}
\end{algorithm}

\paragraph{Check if a set is tight.}
Here, we demonstrate how to utilize max-flow on $\setpairgraph$ to determine if a set $S$ is tight. 
If it is indeed tight, we proceed to remove it from $\activesets$ in Line \ref{line:deactivate_tight_sets} of $\pcsfthree$.

Checking the tightness of a set is straightforward, as outlined in Definition \ref{def:set-tightness}.
To determine if set $S$ is tight, we increase the capacity of the directed edge from $\source$ to $S$ and check if the flow from $\source$ to $\sink$ exceeds $\sum_{S \subset \V} \ys$.
The pseudocode for this function is provided in Algorithm \ref{alg:check_set_tight}.

\begin{algorithm}[ht]
  \caption{Check if set $S \in \activesets$ is tight}
  \label{alg:check_set_tight}
  \hspace*{\algorithmicindent} \textbf{Input:} An undirected graph $\G=(\V, \E, \cc)$ with edge costs $\cc: \E \rightarrow \mathbb{R}_{\ge 0}$, penalties $\pi : \V \times \V \rightarrow \mathbb{R}_{\ge 0}$, \\\hspace*{\algorithmicindent}\hspace{12mm} an instance of \staticcoloring~represented by $\y : 2^{\V} \rightarrow \mathbb{R}_{\ge 0}$, and a set $S \subset \V$. \\
  \hspace*{\algorithmicindent} \textbf{Output:} \textit{True} if set $S$ is tight, \textit{False} otherwise.
  \begin{algorithmic}[1]
    \Function{\checksetistight}{$\G,\pi, \y$, $S$}
        \State $\GF \gets \setpairgraph(\G, \pi, \y)$
        \State Set the capacity of the edge from $\source$ to $S$ in graph $\GF$ equals to $\ys + 1$.      
        \State $\mfout \gets \mincut(\GF, \source, \sink)$
        \If{$\mf > \sum_{S \subset \V} \ys$}
            \State \Return False
        \Else
            \State \Return True
        \EndIf
    \EndFunction
  \end{algorithmic}
\end{algorithm}

\begin{lemma}
\label{lm:check-tight-polynomial}
    Each call to the \checksetistight~function during the algorithm runs in polynomial time.
\end{lemma}
\begin{proof}
    The \checksetistight~function call \mincut~once on the \setpairgraph, whose size remains polynomial throughout the algorithm according to Lemma \ref{lm:size_of_setpairgraph}. 
    Therefore, both the \mincut~and the \checksetistight~function run in polynomial time.
\end{proof}

\paragraph{Reduce the number of tight pairs.}
At the end of $\pcsfthree$, we obtain a final \validstaticcoloring, from which we can derive a corresponding final \dynamiccoloring~which corresponds to a max-flow in $\setpairgraph$.
Next, we present the process of reducing the number of tight pairs in the final \dynamiccoloring, aiming to achieve a \minimaldynamiccoloring.
This step is essential to obtain a 2-approximate solution for PCSF in the next section.

\begin{definition}[Minimal Dynamic Coloring]
\label{def:minimal-dynamic-coloring}
    A \dynamiccoloring~is considered \minimaldynamiccoloring~if there are no pairs $(i, j), (i', j') \in \V \times \V$ and set $S \subset V$ such that pair $(i, j)$ is a tight pair while $(i', j')$ is not a tight pair, $S \odot (i, j)$, $S \odot (i', j')$, and $\ysij > 0$.
\end{definition}

To obtain a \minimaldynamiccoloring, we first check if there exist pairs $(i, j), (i', j')$ and a set $S$ meeting the following criteria: pair $(i, j)$ is tight, pair $(i', j')$ is not tight, $S \odot (i, j)$, $S \odot (i', j')$, and $\ysij > 0$.
If such pairs and set exist, we proceed with the following adjustments. Since $(i', j')$ is not tight, we have $\pi_{i'j'} - \y_{i'j'} > 0$. Additionally, $\ysij > 0$ is assumed.
Therefore, there exists $\epsilon > 0$ such that $\epsilon < \min(\ysij, \pi_{i'j'} - \y_{i'j'})$.
Given that $\epsilon < \ysij \leq \yij$, we can reduce $\ysij$ and $\yij$ by $\epsilon$, while adding $\epsilon$ to $\y_{Si'j'}$ and $\y_{i'j'}$.
Since $\epsilon > 0$, pair $(i, j)$ is no longer tight, and since $\epsilon < \pi_{i'j'} - \y_{i'j'}$ for the previous value of $\y_{i'j'}$, pair $(i', j')$ will not become tight.
It is important to note that the \dynamiccoloring~prior to these changes corresponds to a max-flow in $\setpairgraph$. 
Implementing these adjustments on the flow of edges associated with $\ysij$, $\yij$, $\y_{Si'j'}$, and $\y_{i'j'}$ results in a new max-flow that corresponds to the updated \dynamiccoloring.
This provides an intuition for why the assignment in the new \dynamiccoloring~remains valid.
We illustrate these flow changes in Figure \ref{fig:remove-tight}, and the complete process for achieving a minimal \dynamiccoloring~is described in Algorithm \ref{alg:reduce_tight_pairs}.

\begin{figure}
\centering
\begin{tikzpicture}[scale=0.10]
\def\dx{17}
\def\dy{15}
\def\r{0.7}
\foreach \nd/\x/\y/\tx/\side/\d in {source/0.3/0/source/above/1, s/1/0/S/above/1, v/2/1/(i\comma j)/above/1, vv/2/-1/(i'\comma j')/below/-1, sink/3/0/sink/right/0} {
    \draw[fill=black] (\x*\dx, \y*\dy) circle (\r);
    \node[\side] at (\x*\dx, \y*\dy + \d) {$\tx$};
    \node (\nd) at (\x*\dx, \y*\dy) {};
};
\draw[->] (source) -- (s) node[midway, below] {$y_{S}$};
\draw[->] (s) -- (v) node[midway,sloped,fill=white,above] {\textcolor{red}{$f$}/$\infty$};
\draw[->] (s) -- (vv) node[midway,sloped,fill=white,below] {\textcolor{red}{$f'$}/$\infty$};
\draw[->] (v) -- (sink) node[midway,sloped,fill=white,above] {\textcolor{red}{$\pi_{ij}$}/$\pi_{ij}$};
\draw[->] (vv) -- (sink) node[midway,sloped,fill=white,below] {\textcolor{red}{$f''$}/$\pi_{i'j'}$};

\def\shift{4*\dx}
\foreach \nd/\x/\y/\tx/\side/\d in {source/0.5/0/source/above/1, s/1/0/S/above/1, v/2/1/(i\comma j)/above/1, vv/2/-1/(i'\comma j')/below/-1, sink/3/0/sink/right/0} {
    \draw[fill=black] (\shift+\x*\dx, \y*\dy) circle (\r);
    \node[\side] at (\shift+\x*\dx, \y*\dy + \d) {$\tx$};
    \node (a-\nd) at (\shift+\x*\dx, \y*\dy) {};
};
\draw[->] (a-source) -- (a-s);
\draw[->] (a-s) -- (a-v) node[midway,sloped,fill=white,above] {\textcolor{red}{$f-\epsilon$}};
\draw[->] (a-s) -- (a-vv) node[midway,sloped,fill=white,below] {\textcolor{red}{$f'+\epsilon$}};
\draw[->] (a-v) -- (a-sink) node[midway,sloped,fill=white,above] {\textcolor{red}{$\pi_{ij}-\epsilon$}};
\draw[->] (a-vv) -- (a-sink) node[midway,sloped,fill=white,below] {\textcolor{red}{$f''+\epsilon$}};
\end{tikzpicture}
\caption{By choosing a small positive value $\epsilon < \min(f, \pi_{i'j'}-f'')$, we can remove one tight pair. The \textcolor{red}{red} variables  represent the amounts of flow on each edge, while the black variables represent their capacity.}
\label{fig:remove-tight}
\end{figure}
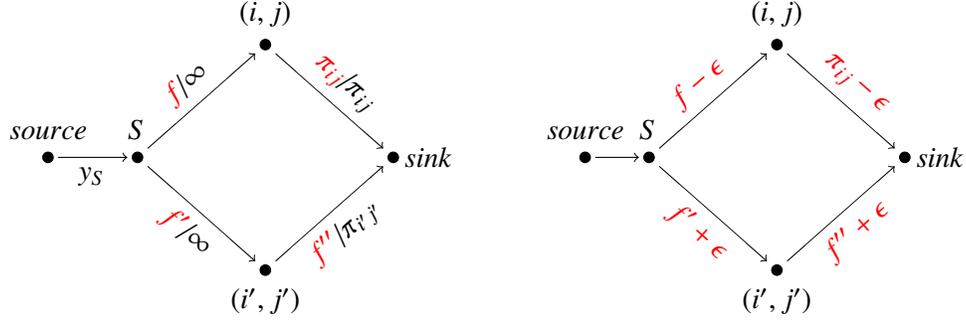

After applying these adjustments, the number of tight pairs is reduced by one.
If there are no tight pairs where their tightness can be removed through this operation, the result is a \minimaldynamiccoloring.
Given that the number of tight pairs is at most $n^2$, and after each operation, the number of tight pairs is reduced by one, after a maximum of $n^2$ iterations in $\reducetightpairs$, the number of tight pairs becomes minimal.

Considering that the number of iterations in the $\reducetightpairs$ function is polynomial and the $\mincut$ operation on the \setpairgraph~in Line \ref{line:reduce-mf} has a polynomial size (Lemma \ref{lm:size_of_setpairgraph}), it can be concluded that the runtime of $\reducetightpairs$ is polynomial.

\begin{corollary}
\label{col:reduce-polynomial}
The \reducetightpairs~function runs in polynomial time.
\end{corollary}

\begin{algorithm}[ht]
  \caption{Reduce the number of tight pairs}
  \label{alg:reduce_tight_pairs}
  \hspace*{\algorithmicindent} \textbf{Input:} An undirected graph $\G=(\V, \E, \cc)$ with edge costs $\cc: \E \rightarrow \mathbb{R}_{\ge 0}$ and penalties $\pi : \V \times \V \rightarrow \mathbb{R}_{\ge 0}$. \\
  \hspace*{\algorithmicindent} \textbf{Output:} The set $\Q$ of tight pairs for which we will pay penalties.
  \begin{algorithmic}[1]
    \Function{\reducetightpairs}{$\G,\pi, \y$}
        \State $\GF \gets \setpairgraph(\G, \pi, \y)$
        \State $\mfout \gets \mincut(\GF, \source, \sink)$
        \label{line:reduce-mf}
        \State Let $\ysij \gets \f(e)$ for each set $S$ and each pair $(i, j)$ that $S \odot (i, j)$, where $e$ is the edge from $S$ to $(i, j)$.
        \State Let $\yij \gets \f(e)$ for each pair $(i, j)$, where $e$ is the edge from $(i, j)$ to $\sink$.
        \For{$(i, j), (i', j') \in \V \times \V$ and $S \subset \V$ that $S \odot (i, j)$, $S \odot (i', j')$, $\yij = \pij$, $\y_{i'j'} < \pi_{i'j'}$, and $\y_{S'ij} > 0$}
        \label{line:reduce-loop}
            \State $\ysij \gets \ysij - \epsilon$
            \State $\yij \gets \yij - \epsilon$
            \State $\y_{Si'j'} \gets \y_{Si'j'} + \epsilon$
            \State $\y_{i'j'} \gets \y_{i'j'} + \epsilon$
        \EndFor
        \State Let $\Q \gets \{(i, j)\in \V \times \V: \sum_{S: S \odot (i, j)} \ysij = \pij\}$
    \State \Return $\Q$
    \EndFunction
  \end{algorithmic}
\end{algorithm}

Finally, after obtaining a \minimaldynamiccoloring, we consider it as our final \dynamiccoloring, which will be used in the analysis presented in Section \ref{sec:2_apx_analysis}.
\begin{corollary}
The final \dynamiccoloring~obtained at the end of procedure $\pcsfthree$ is a \minimaldynamiccoloring.
\end{corollary}

Furthermore, in a \minimaldynamiccoloring, we establish the following lemma, which is necessary for the analysis presented in the next section.

\begin{lemma}
    \label{lemma:minimal-cut-tight}
    In a \minimaldynamiccoloring, if a set $S \subset \V$ cuts a tight pair $(i, j) \in \V \times \V$ with $\ysij > 0$, then all pairs $(i', j')$ satisfying $S \odot (i', j')$ are also tight.
\end{lemma}
\begin{proof}
    Assume there exists a pair $(i', j')$ satisfying $S \odot (i', j')$ that is not tight. This implies that the pairs $(i, j)$ and $(i', j')$, along with set $S$, contradict the definition of \minimaldynamiccoloring~(Definition \ref{def:minimal-dynamic-coloring}).
\end{proof}


\subsection{Analysis} 
In this section, we demonstrate the validity of our algorithm's output for the given PCSF instance. 
We also present some lemmas that are useful for proving the approximation factor of $\pcsfthree$. 
However, we do not explicitly prove the approximation factor of $\pcsfthree$ in this section, as it is not crucial for our main result. 
Nonetheless, one can easily conclude the $3$-approximation factor of $\pcsfthree$ using Lemmas \ref{lemma:opt_lb},  \ref{lemma:pcsf3-ub}, and \ref{lm:all_vb_cut_opt} provided in the next section.
Additionally, in Lemma \ref{lm:pcsfthree_polynomial}, we show that $\pcsfthree$ has a polynomial time complexity. 
The lemmas provided in this section are also necessary for the analysis of our 2-approximation algorithm, which is presented in the next section.

To conclude the correctness of our algorithm, it is crucial to show that our algorithm pays penalties for all pairs that are not connected in $\Fp$. 
In other words, every pair that is not tight will be connected in $\Fp$. 
This ensures that by paying the penalties for tight pairs and the cost of edges in $\Fp$, we obtain a feasible solution.

To prove this, we introduce some auxiliary lemmas. 
First, in Lemma \ref{lemma:rem-tight}, we demonstrate that when a set becomes tight during $\pcsfthree$, it remains tight until the end of the algorithm. 
This lemma is essential because if a set becomes tight and is subsequently removed from the active sets, but then becomes non-tight again, it implies that some pairs could contribute to the coloring in the \dynamiccoloring, but their colors may no longer be utilized.

Furthermore, in Lemma \ref{lemma:c-contains-tight}, we establish that every connected component of $\F$ at the end of $\pcsfthree$ is a tight set. This provides additional evidence that the algorithm produces a valid solution.

Finally, we use these lemmas to prove the validity of the solution produced by $\pcsfthree$ in Lemma \ref{lm:F_connect_not_tight_pairs}.

Let $\SC$ be the cut in $\setpairgraph$ that separates $\source$ from the other vertices.

\begin{lemma}
\label{lemma:set-is-in-min-cut}
    At every moment of $\pcsfthree$, in the $\setpairgraph$ representation corresponding to the \staticcoloring~of that moment, $\SC$ is a minimum cut between $\source$ and $\sink$.
\end{lemma}
\begin{proof}
    According to Lemma \ref{lemma:always-valid}, the \staticcoloring~is always valid during $\pcsfthree$.
    Moreover, based on Corollary \ref{col:valid_dynamic_coloring}, when the \staticcoloring~is valid, $\SC$ represents a minimum cut.
\end{proof}

\begin{lemma}
\label{lemma:min-cut-wo-tight}
    A set $S \subset V$ is tight if and only if there exists a minimum cut between $\source$ and $\sink$ in \setpairgraph~representation of a valid \dynamiccoloring~that does not contain the edge $e$ from $\source$ to $S$.
\end{lemma}
\begin{proof}
    Using contradiction, let's assume that $S$ is tight and all minimum cuts contain the edge $e$.
    Let $\epsilon>0$ be the difference between the weight of the minimum cut and the first cut whose weight is greater than the minimum cut.
    By increasing the capacity of the edge $e$ by $\epsilon$, the weight of any minimum cut increases by a positive value $\epsilon$, as well as the maximum flow.
    This implies that we can increase $\ys$ and still maintain a \validstaticcoloring.
    Therefore, based on the definition of set tightness (Definition \ref{def:set-tightness}), we can conclude that set $S$ is not tight.
    This contradicts the assumption of the tightness of $S$ and proves that there exists a minimum cut that does not contain the edge $e$.
    
    Furthermore, if we have a minimum cut that does not contain edge $e$, increasing the capacity of $e$ does not affect the value of that minimum cut and respectfully the maximum flow. 
    By using Lemma \ref{lm:valid_dynamic_coloring}, we conclude that increasing $\ys$ would result in an invalid \staticcoloring.
    Therefore, based on Definition \ref{def:set-tightness}, we can conclude that set $S$ is tight.
\end{proof}

\begin{lemma}
    \label{lemma:rem-tight}
    Once a set $S$ becomes tight, it remains tight throughout the algorithm.
\end{lemma}
\begin{proof}
    According to Lemma \ref{lemma:set-is-in-min-cut}, $\SC$ is always a minimum cut.
    Let us assume that at time $t$, the set $S$ becomes tight.
    Based on Lemma \ref{lemma:min-cut-wo-tight}, there exists a minimum cut $C_S$ that has $S$ on the side of $\source$.
    Therefore, at time $t$, $\SC$ and $C_S$ have the same weight.
    Now, let us consider a contradiction by assuming that there is a time $t' > t$ when $S$ is not tight.
    The only difference between $\setpairgraph$ at time $t$ and time $t'$ is the increased capacity of some edges between $\source$ and sets $S' \subset \V$.
    Let us assume that the total increase in all $\y_{S'}$ from time $t$ to $t'$ is $d$.
    Since all of these edges are part of the cut $\SC$, the weight of the cut $\SC$ is increased by $d$.
    Furthermore, since the total capacity of all edges in $\setpairgraph$ from time $t$ to $t'$ has increased by $d$, the weight of $C_S$ cannot have increased by more than $d$.
    That means, the weight of $C_S$ cannot exceed the weight of $\SC$ at time $t'$.
    Since $\SC$ is a minimum cut at time $t'$ according to Lemma \ref{lemma:set-is-in-min-cut}, we can conclude that $C_S$ remains a minimum cut at time $t'$.
    Therefore, based on Lemma \ref{lemma:min-cut-wo-tight}, the set $S$ is still tight at time $t'$, which contradicts the assumption that it is not tight.
\end{proof}

\begin{lemma}
\label{lemma:c-contains-tight}
    At the end of \pcsfthree, all remaining sets in $\currentsets$ are tight. 
\end{lemma}
\begin{proof}
    In Line \ref{line:init-acts} of the algorithm, both \activesets~and \currentsets~are initialized with the same set of sets. 
    Additionally, in Lines \ref{line:update-c} and \ref{line:update-as}, the same sets are removed from \activesets~and \currentsets~or added to both data structures. 
    The only difference occurs in Line \ref{line:deactivate_tight_sets}, where tight sets are removed from \activesets~but not from \currentsets. 
    Given Lemma~\ref{lemma:rem-tight}, these sets are tight at the end of $\pcsfthree$.
    Therefore, at the end of the algorithm, since there are no sets remaining in \activesets, all sets in \currentsets~are tight.
\end{proof}

\begin{lemma}
\label{lm:F_connect_not_tight_pairs}
    After executing $\reducetightpairs$, the endpoints of any pair that is not tight will be connected in the forest $\Fp$.
\end{lemma}
\begin{proof}
    The forest $\Fp$ is obtained by removing redundant edges from $\F$, which are edges that are not part of a path between the endpoints of a pair that is not tight. 
    Hence, we only need to show that every pair that is not tight is connected in $\F$.
    Let us assume, for the sake of contradiction, that there exists a pair $(i, j)$ that is not tight and the endpoints $i$ and $j$ are not connected in $\F$. 
    Consider the set $S \in \currentsets$ at the end of the algorithm that contains $i$. Since $i$ and $j$ are not connected in $\F$, and $S$ is a connected component of $\F$, it follows that $S$ cuts the pair $(i, j)$. 
    According to Lemma \ref{lemma:c-contains-tight}, $S$ is a tight set.
    This contradicts Lemma \ref{lemma:tight-set-tight-pairs} because we have a tight set $S$ such that $S \odot (i, j)$ is not tight. 
    Therefore, our assumption is false, and every pair that is not tight is connected in $\F$. 
    As a result, after executing $\reducetightpairs$, the endpoints of any pair that is not tight will be connected in the forest $\Fp$.
\end{proof}

Now we will prove that the running time of $\pcsfthree$ is polynomial.

\begin{lemma}
\label{lm:pcsfthree_polynomial}
    For instance $\I$, the runtime of $\pcsfthree$ is polynomial.
\end{lemma}
\begin{proof}
    We know that $\Delta_e$ denotes the time it takes for at least one new edge to be fully colored according to Corollary \ref{col:deltae-fully-colored}, and $\Delta_p$ signifies the time required for at least one active set to be deactivated based on the maximality of $\Deltap$ demonstrated in Lemma \ref{lemma:finddeltap-max}. During each iteration of the while loop at Line \ref{line:pcsf3-while}, it is guaranteed that at least one of these events takes place.

    If an edge becomes fully colored, it results in the merging of two sets into one in $\currentsets$. 
    As a result, two sets are removed and one set is added at Line \ref{line:update-c}, leading to a decrease in the size of $\currentsets$.
    Alternatively, if an active set is deactivated, it is removed from $\activesets$ at Line \ref{line:deactivate_tight_sets}, which leads to a decrease in the size of $\activesets$.
    It is important to note that the number of active sets in $\activesets$ does not increase at Line \ref{line:update-as} (it either decreases by one or remains the same).
    From this, we can conclude that after each iteration of the while loop, either the number of active sets in $\activesets$ decreases by at least one, or the number of sets in $\currentsets$ decreases by one, or both events occur. Since both $\activesets$ and $\currentsets$ initially contain $n$ elements, the while loop can iterate for a maximum of $2n$ times.
    
    In each iteration, we perform the following operations with polynomial runtime: $\finddeltae$, which is polynomial due to Lemma \ref{lm:finddeltae-polynomial}; $\finddeltap$, which is polynomial according to Corollary \ref{col:finddeltap-polynomial}; iterating through active sets to extend the \staticcoloring, which is polynomial based on the size of $\activesets$; iterating through edges to update active sets if they fully color edges, which is polynomial; and checking if each active set is tight using $\checksetistight$, which is polynomial according to Lemma \ref{lm:check-tight-polynomial}.
    
    In the end, we also run $\reducetightpairs$, which is polynomial according to Corollary \ref{col:reduce-polynomial}.
    
    Therefore, we can conclude that $\pcsfthree$ runs in polynomial time. 
\end{proof}


\section{The Iterative Algorithm}
\label{sec:2_apx_alg}
In this section, we present our iterative algorithm which uses the \pcsfthree~procedure from Algorithm \ref{alg:pcsf3} as a building block.
We then provide a proof of its $2$-approximation guarantee in Section \ref{sec:2_apx_analysis}.
Finally, in Section~\ref{sec:3-2-n-approax}, we provide a brief overview of a more refined analysis to establish a $(2-\frac{1}{n})$-approximation for an $n$ vertex input graph.

Our algorithm, described in Algorithm \ref{alg:rpcsf}, considers two solutions for the given PCSF instance $\I$.
The first solution, denoted as $(\Q_1,\Fp_1)$, is obtained by invoking the \pcsfthree~procedure (Line \ref{line:get_pcsf3_output}). If the total penalty of this solution, $\pi(\Q_1)$, is equal to $0$, the algorithm returns it immediately as the solution.

Otherwise, a second solution, denoted as $(\Q_2,\Fp_2)$, is obtained through a recursive call on a simplified instance $\R$.
The simplified instance is created by adjusting penalties: penalties are limited to pairs that Algorithm \ref{alg:pcsf3} does not pay, and the penalties for other pairs are set to $0$ (Lines \ref{line:initial_pi}-\ref{line:construct_R}).
Essentially, we assume that pairs whose penalties are paid in the first solution will indeed be paid, and our objective is to find a solution for the remaining pair connection demands. 
We note that setting the penalties for these pairs to $0$ guarantees their inclusion in $\Q_2$. 
This is because $\Q_2$ represents the set of tight pairs for a subsequent invocation of $\pcsfthree$, and any pair with a penalty of $0$ is trivially tight.

To compare the two solutions, the algorithm computes the values $cost_1=c(\Fp_1)+\pi(\Q_1)$ and $cost_2=c(\Fp_2)+\pi(\Q_2)$, which represent the costs of the solutions (Lines \ref{line:cost1} and \ref{line:cost2}). In the final step, the algorithm simply selects and returns the solution with the lower cost. 

\begin{algorithm}[ht]
  \caption{Iterative PCSF algorithm}
  \label{alg:rpcsf}
  \hspace*{\algorithmicindent} \textbf{Input:} An undirected graph $\G=(\V, \E, \cc)$ with edge costs $\cc: \E \rightarrow \mathbb{R}_{\ge 0}$ and penalties $\pi : \V \times \V \rightarrow \mathbb{R}_{\ge 0}$. \\
  \hspace*{\algorithmicindent} \textbf{Output:} A set of pairs $\Q$ with a forest $\Fp$ that connects the endpoints of every pair $(i, j) \notin \Q$.
  \begin{algorithmic}[1]
    \Procedure{\rpcsf}{$\I=(\G,\ \pi)$}
        \State $(\Q_1,\Fp_1)\gets \pcsfthree(\I)$ \label{line:get_pcsf3_output}
        \If{$\pi(\Q_1)=0$} \label{line:check_Q_1_is_empty}
        \State\Return $(\Q_1,\Fp_1)$ \label{line:return-early}
        \EndIf
        \State $cost_1\gets c(\Fp_1)+\pi(\Q_1)$
        \label{line:cost1}
        \State Initialize $\pi'$ as a new all-zero penalty vector
        \label{line:initial_pi}
        \For{$(i,j) \in \V \times \V$}
            \If{$(i,j) \in \Q_1$}
            \State $\pi'_{ij}\gets0$ \label{line:pip_set_zero}
            \Else 
            \State $\pi'_{ij}\gets\pij$ \label{line:set_pip}
            \EndIf
        \EndFor

        \State Construct instance $\R$ of the PCSF problem consisting of $\G$ and $\pi'$
        \label{line:construct_R}
        \State $(\Q_2,\Fp_2) \gets \rpcsf(\R)$ \label{line:get_recursive_output}
        \State $cost_2\gets c(\Fp_2)+\pi(\Q_2)$
        \label{line:cost2} 
        \If{$cost_1 \leq cost_2$}
        \label{line:return-min}
        \State\Return $(\Q_1,\Fp_1)$
        \Else
        \State \Return $(\Q_2,\Fp_2)$ 
        \EndIf
    \EndProcedure
  \end{algorithmic}
\end{algorithm}


\subsection{Analysis}
\label{sec:2_apx_analysis}
We now analyze the approximation guarantee of Algorithm \ref{alg:rpcsf}. In the following, we consider an arbitrary instance $\I = (\G, \pi)$ of the PCSF problem, and analyze the solutions found by the $\rpcsf$ algorithm.
In our analysis, we focus on \textbf{the first call} of $\rpcsf$. 
By the output of $\pcsfthree$, we refer to the result of the first call of $\pcsfthree$ on instance $\I$ at Line \ref{line:get_pcsf3_output}. 
Similarly, when we mention the output of the recursive call, we are referring to the output of $\rpcsf$ on instance $\R$ at Line \ref{line:get_recursive_output}.
We compare the output of $\rpcsf$ on $\I$, which is the minimum of the output of $\pcsfthree$ and the output of the recursive call, with an optimal solution $\OPT$ of the instance $\I$. 
We denote the forest selected in $\OPT$ as $\fopt$ and use $\qopt$ to refer to the set of pairs not connected in $\fopt$, for which $\OPT$ pays the penalties. Then, the cost of $\OPT$ is given by $cost(\OPT)=c(\fopt)+\pi(\qopt)$.

In the following, when we refer to coloring, we specifically mean the coloring performed in the first call of $\pcsfthree$ on instance $\I$. In particular, when we mention \dynamiccoloring, we are referring to the final \dynamiccoloring~of the first call of $\pcsfthree$ on instance $\I$. The values $\ys$, $\ysij$, and $\yij$ used in the analysis all refer to the corresponding values in the final \staticcoloring~and \dynamiccoloring.

\begin{definition}
\label{def:tbl}
For an instance $I$, we define four sets to categorize the pairs based on their connectivity in both the optimal solution $\OPT$ of $I$ and the result of $\pcsfthree(I)$, denoted as $(\Q_1, \Fp_1)$:

\begin{itemize}
    \item Set $\A$ contains pairs $(i, j)$ that are connected in the optimal solution and are not in the set $\Q_1$ returned by \pcsfthree.
    \item Set $\B$ contains pairs $(i, j)$ that are connected in the optimal solution and are in the set $\Q_1$ returned by \pcsfthree.
    \item Set $\C$ contains pairs $(i, j)$ that are not connected in the optimal solution and are not in the set $\Q_1$ returned by \pcsfthree.
    \item Set $\D$ contains pairs $(i, j)$ that are not connected in the optimal solution and are in the set $\Q_1$ returned by \pcsfthree.
\end{itemize}
 
Based on the final \dynamiccoloring~of $\pcsfthree(I)$, we define the following values to represent the total duration of coloring with pairs in these sets.
\begin{align*}
    &\va = \sum_{(i,j)\in \A} \yij,
    &\vb = \sum_{(i,j)\in \B} \yij\\
    &\vc = \sum_{(i,j)\in \C} \yij,
    &\vd = \sum_{(i,j)\in \D} \yij\\
\end{align*}
The following table illustrates the connectivity status of pairs in each set:
\begin{center}
\begin{tabular}{cccc}
\toprule
& & \multicolumn{2}{c}{\pcsfthree} \\
\cmidrule{3-4}
& & Connect & Penalty \\
\midrule
\multirow{2}{*}{Optimal Solution} & Connect & \A & \B \\
& Penalty & \C & \D \\
\bottomrule
\end{tabular}
\end{center}
\bigskip
\end{definition}


So far, we have classified pairs into four categories. Now, we categorize the coloring moments involving pairs in set $\B$ into two types: those that color exactly one edge of the optimal solution, and those that color at least two edges of the optimal solution.
Since pairs in $\B$ are connected in the optimal solution, they are guaranteed to color at least one edge of the optimal solution during their coloring moments.
Furthermore, we allocate the value of $\vb$ between $\vba$ and $\vbb$ based on this categorization.

\begin{definition}[Single-edge and multi-edge sets]
\label{def:multi-ring}
For an instance $\I$, we define a set $S \subset \V$ as a single-edge set if it cuts exactly one edge of $\OPT$, i.e., $\dopt(S) = 1$, and as a multi-edge set if it cuts at least two edges of $\OPT$, i.e., $\dopt(S) > 1$.
Let $\vba$ represent the duration of coloring with pairs in $\B$ in \dynamiccoloring~that corresponds to coloring with single-edge sets in \staticcoloring.
Similarly, let $\vbb$ represent the duration of coloring with pairs in $\B$ in \dynamiccoloring~that corresponds to coloring with multi-edge sets in \staticcoloring. 
These values are formally defined as follows:
\begin{align*}
    \vba =\sum_{(i,j) \in \B} \sum_{\substack{S:S\odot(i,j),\\ \dopt(S)=1}} \ysij\\
    \vbb =\sum_{(i,j) \in \B} \sum_{\substack{S:S\odot(i,j),\\ \dopt(S)>1}} \ysij\text{.}
\end{align*}
\end{definition}

Figure \ref{fig:single_cut_ring} displays a single-edge set on the left and a multi-edge set on the right.

\begin{figure}
\begin{center}
\begin{tikzpicture}[scale = 0.15]
\def\dx{9}
\def\dy{5}
\def\r{0.7}
\foreach \x/\y/\R in {-1/-2/4, 2.5/-1/10} {
    \draw[double=gray!30,double distance=4,thick] (\x*\dx, \y*\dy) circle (\R);
};
\foreach \x/\y in {0/0, 2/0, -1/-2, -0.5/1, 2.5/-1, 3/2} {
    \draw[fill=black] (\x*\dx, \y*\dy) circle (\r);
};
\foreach \x/\y/\X/\Y in {0/0/2/0, 0/0/-1/-2, 0/0/-0.5/1, 2/0/2.5/-1, 2/0/3/2} {
    \draw (\x*\dx, \y*\dy) -- (\X*\dx, \Y*\dy);
};
\end{tikzpicture}
\end{center}
\caption{A comparison between a single-cut set (left) and a multi-cut set (right).} \label{fig:single_cut_ring}
\end{figure}
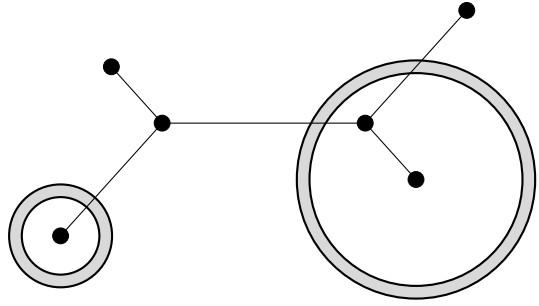

\begin{lemma}
\label{lm:all_vb_cut_opt}
    For an instance $\I$, we have $\vba+\vbb=\vb$.
\end{lemma}
\begin{proof}
    Since pairs in $\B$ are connected by the optimal solution $\OPT$, any set $S$ cutting a pair in $\B$ must cut at least one edge of $\OPT$. 
    Therefore, $S$ is either a single-edge set or a multi-edge set. 
    Hence, we have $\vba+\vbb=\vb$.
\end{proof}

Now, we use these definitions and categorizations to analyze our algorithm.
All of the following lemmas are based on the assumption that $\rpcsf$ is executed on instance $\I$.
First, in Lemma \ref{lemma:opt_lb}, we provide a lower bound on the cost of the optimal solution, which is $cost(\OPT) \ge \va + \vb + \vbb + \vc + \vd$.
Next, in Lemma \ref{lemma:pcsf3-ub}, we present an upper bound on the output of $\pcsfthree(\I)$, which is $cost_1 \le 2\va + 2\vc + 3\vb + 3\vd$. 
Moreover, in Lemma \ref{lemma:pcsf3-ub-opt}, we show that this value is at most $2cost(\OPT) + \vba - \vbb + \vd$.

Next, we want to bound the output of the recursive call within $\rpcsf$.
In Lemma~\ref{lemma:optr-up}, we initially proof that $cost(\OPTR) \le cost(\OPT) - \vd - \vba$, where $cost(\OPTR)$ represents the cost of the optimal solution for the instance $\R$ defined at Line~\ref{line:construct_R}.
Finally, in Theorem~\ref{thm:main_theorem_with_proof}, we employ induction to demonstrate that $cost(\rpcsf) \le 2cost(\OPT)$. Here, $cost(\rpcsf)$ denotes the cost of the output produced by $\rpcsf$ on instance $\I$.
To accomplish this, we use the same induction to bound the cost of the solution obtained through the recursive call at Line~\ref{line:cost2} by $cost_2 \le 2cost(\OPTR) + \vb + \vd$, and by utilizing Lemma \ref{lemma:optr-up}, we can then conclude that $cost_2 \le 2cost(\OPT) - \vba + \vbb - \vd$.
Taking the average of $cost_1$ and $cost_2$ results in a value that is at most $2cost(\OPT)$.
Consequently, the minimum of these two values, corresponding to the cost of $\rpcsf(\I)$, is at most $2cost(\OPT)$.

\begin{lemma}
\label{lemma:opt_lb}
    For an instance $\I$, We can derive a lower bound for the cost of the optimal solution $\OPT$ as follows:
    $$cost(\OPT) \ge \va + \vb + \vbb + \vc + \vd\text{.}$$ 
\end{lemma}
\begin{proof}
The optimal solution pays penalties for pairs with labels $\C$ and $\D$ as it does not connect them.
Since $\yij \le \pij$ for each pair $(i,j)$, we can lower bound the penalty paid by $\OPT$ as
\begin{align*}
    \pi(\qopt) \ge \sum_{(i,j)\in (\C \cup \D)} \pij \ge \sum_{(i,j)\in (\C \cup \D)} \yij = \vc + \vd. 
\end{align*}

Now, we want to bound the cost of the forest in the optimal solution by $\va + \vb + \vbb$.
First, it is important to note that each part of an edge will be colored at most once. During the execution of the \staticcoloring, each active set $S$ colors the uncolored parts of all its cutting edges. 
Therefore, when $S$ is an active set, it colors parts of exactly $\dopt(S)$ edges of the optimal solution.
Based on this observation, we can bound the total cost of the edges in the optimal solution by considering the amount of coloring applied to these edges.
\begin{align*}
    \cc(\fopt) &\ge \sum_{S \subset \V} \dopt(S)\cdot\ys \\
    &= \sum_{S \subset \V} \sum_{(i, j): S\odot(i, j)} \dopt(S)\cdot\ysij \tag{$\ys = \sum_{(i, j): S\odot (i, j)}\ysij$}\\ 
    &= \sum_{(i, j) \in \V \times \V} \sum_{S: S \odot (i, j)}  \dopt(S)\cdot\ysij \tag{change the order of summations}\\
    &\ge\sum_{(i,j)\in\A} \sum_{S\odot(i,j)} \dopt(S)\cdot\ysij + \sum_{(i,j)\in\B} \sum_{S\odot(i,j)} \dopt(S)\cdot\ysij\text{.} \tag{$\A \cap \B = \emptyset$}
\end{align*}

For each pair $(i, j) \in (\A \cup \B)$, we know that in the optimal solution $\OPT$, the endpoints of $(i, j)$ are connected. This implies that for every set $S$ satisfying $S \odot (i, j)$, the set $S$ cuts the forest of $\OPT$, i.e., $\dopt(S) \ge 1$.
Based on this observation, we bound the two terms in the summation above separately.
For pairs in $\A$, we have
\begin{align*}
    \sum_{(i,j)\in\A} \sum_{S\odot(i,j)} \dopt(S)\cdot\ysij \ge \sum_{(i,j)\in\A} \sum_{S\odot(i,j)} \ysij = \sum_{(i,j)\in \A} \yij = \va.
\end{align*}

For pairs in $\B$, we have
\begin{align*}
    \sum_{(i,j)\in\B} \sum_{S\odot(i,j)} \dopt(S)\cdot\ysij &=\sum_{(i,j) \in \B} \sum_{\substack{S\odot(i,j),\\ \dopt(S)=1}} \dopt(S)\cdot\ysij + \sum_{(i,j) \in \B} \sum_{\substack{S\odot(i,j),\\ \dopt(S)>1}} \dopt(S)\cdot\ysij\\
    &\ge\sum_{(i,j) \in \B} \sum_{\substack{S\odot(i,j),\\ \dopt(S)=1}} \ysij + \sum_{(i,j) \in \B} \sum_{\substack{S\odot(i,j),\\ \dopt(S)>1}} 2\ysij\\
    &= \vba + 2\vbb\\
    &= \vb + \vbb \tag{Lemma \ref{lm:all_vb_cut_opt}}\text{.}
\end{align*}

Summing up all the components, we have:
$$cost(\OPT) = \cc(\fopt) + \pi(\qopt) \ge \va + \vb + \vbb + \vc + \vd$$
\end{proof}

\begin{lemma}
\label{lemma:single-cut}
    Let $\F$ be an arbitrary forest and $S$ be a subset of vertices in $\F$. If $S$ cuts only one edge $e$ in $\F$,  then removing this edge will only disconnect pairs of vertices cut by $S$.
\end{lemma}
\begin{proof}
    Consider a pair $(i,j)$ that is disconnected by removing $e$. 
    This pair must be connected in forest $\F$, so there is a unique simple path between $i$ and $j$ in $\F$. 
    This path must include edge $e$, as otherwise, the pair would remain connected after removing $e$. 
    Let the endpoints of $e$ be $u$ and $v$, where $u \in S$ and $v\not\in S$. 
    Without loss of generality, assume that $i$ is the endpoint of the path that is closer to $u$ than $v$. Then $i$ is connected to $u$ through the edges in the path other than $e$. As these edges are not cut by $S$ and $u \in S$, it follows that $i$ must also be in $S$. 
    Similarly, it can be shown that $j$ is not in $S$. Therefore, $S$ cuts the pair $(i,j)$.
\end{proof}

\begin{lemma}
\label{lemma:pcsf3-ub}
    For an instance $\I$, during the first iteration of $\rpcsf(\I)$ where $\pcsfthree(\I)$ is invoked, we can establish an upper bound on the output of $\pcsfthree$ as follows:
    $$cost_1 \le 2\va + 2\vc + 3\vb + 3\vd\text{.}$$
\end{lemma}
\begin{proof}
Since $cost_1$ is the total cost of $\pcsfthree(\I)$, we should bound $\pi(\Q_1) + \cc(\Fp_1)$.
First, let's observe that $\pcsfthree$ pays the penalty for exactly the pairs $(i, j)$ in $\B \cup \D$, where $\B \cup \D = \Q_1$. 
Since every pair in $\Q_1$ is tight, we have $\pij = \yij$ for these pairs.
Therefore, the total penalty paid by $\pcsfthree$ can be bounded by
$$\pi(\Q_1) = \sum_{(i, j) \in (\B \cup \D)} \pij = \sum_{(i, j) \in (\B \cup \D)} \yij = \vb + \vd\text{.}$$

Now, it suffices to show that $c(\Fp_1) \le 2(\va+\vb+\vc+\vd)$. 
We can prove this similarly to the proof presented by Goemmans and Williamson in \cite{DBLP:journals/siamcomp/GoemansW95}.
Since each pair belongs to exactly one of the sets $\A$, $\B$, $\C$, and $\D$, we can observe that
$$\va+\vb+\vc+\vd = \sum_{(i, j) \in \V \times \V} \yij = \sum_{S\subset \V} \ys\text{.}$$
Therefore, our goal is to prove that the cost of $\Fp_1$ is at most $2\sum_{S\subset \V} \ys$  using properties of \staticcoloring.
To achieve this, we show that the portion of edges in $\Fp_1$ colored during each step of $\pcsfthree$ is at most twice the total increase in the $\ys$ values during that step. 
Since every edge in the forest $\Fp_1$ is fully colored by $\pcsfthree$, this will establish the desired inequality.

Now, let's consider a specific step of the procedure $\pcsfthree$ where the $\ys$ values of the active sets in $\activesets$ are increased by $\Delta$. In this step, the total proportion of edges in $\Fp_1$ that are colored by an active set $S$ is $\Delta d_{\Fp_1}(S)$. Therefore, we want to prove that
$$
\Delta \sum_{S\in\activesets} d_{\Fp_1}(S) \le 2\Delta\cdot|\activesets|\text{,}
$$
where the left-hand side represents the length of coloring on the edges of $\Fp_1$ in this step, while the right-hand side represents twice the total increase in $\ys$ values.

Consider the graph $H$ formed from $\Fp_1$ by contracting each connected component in $\currentsets$ at this step in the algorithm. As the edges of forest $\F$ at this step and $\Fp_1$ are a subset of the forest $\F$ at the end of $\pcsfthree$, the graph $H$ should be a forest. 
If $H$ contains a cycle, it contradicts the fact that $\F$ at the end of $\pcsfthree$ is a forest.

In forest $H$, each vertex represents a set $S \in \currentsets$, and the neighboring edges of this vertex are exactly the edges in $\delta(S)\cap\Fp_1$. 
We refer to the vertices representing active sets as active vertices, and the vertices representing inactive sets as inactive vertices. 
To simplify the analysis, we remove any isolated inactive vertices from $H$.

Now, let's focus on the inactive vertices in $H$. 
Each inactive vertex must have a degree of at least $2$ in $H$. 
Otherwise, if an inactive vertex $v$ has a degree of $1$, consider the only edge in $H$ connected to this vertex. 
For this edge not to be removed in the final step of Algorithm \ref{alg:pcsf3} at Line \ref{line:create_F'}, there must exist a pair outside of $\Q_1$ that would be disconnected after deleting this edge. 
However, since vertex $v$ is inactive, its corresponding set $S$ becomes tight before this step. 
According to Lemma \ref{lemma:rem-tight}, $S$ will remain tight afterward. 
As a result, by Lemma \ref{lemma:tight-set-tight-pairs}, any pair cut by $S$ will also be tight in the final coloring and will be included in $\Q_1$. 
By applying Lemma \ref{lemma:single-cut}, we can conclude that the only pairs disconnected by removing this edge would be the pairs cut by $S$, which we have shown to be in $\Q_1$. 
Therefore, an inactive vertex cannot have a degree of $1$, and all inactive vertices in $H$ have a degree of at least $2$. 
Let $V_a$ and $V_i$ represent the sets of active and inactive vertices in $H$, respectively.
We have
\begin{align*}
        \sum_{S\in\activesets} d_{\Fp_1}(S) 
        &= 
        \sum_{v\in V_a} d_H(v)\\
        &=\sum_{v\in V_a\cup V_i} d_H(v) - \sum_{v\in V_i} d_H(v) \\
        &\le 2(\lvert V_a \rvert + \lvert V_i\rvert)-\sum_{v\in V_i} d_H(v) \tag{$H$ is a forest}\\
        &\le 2(\lvert V_a \rvert + \lvert V_i\rvert)- 2\lvert V_i \rvert \tag{$d_H(v)\ge 2$ for $v\in V_i$}\\
        &\le 2(\lvert V_a\rvert) = 2\lvert\activesets\rvert\text{.}
\end{align*}
This completes the proof.
\end{proof}

\begin{lemma}
\label{lemma:pcsf3-ub-opt}
    For an instance $\I$, during the first iteration of $\rpcsf(\I)$ where $\pcsfthree(\I)$ is invoked, we can establish an upper bound on the output of $\pcsfthree$ as follows:
    $$cost_1 \le 2cost(\OPT) + \vba - \vbb + \vd$$
\end{lemma}
\begin{proof}
We can readily prove this by referring to the previous lemmas.
\begin{align*}
cost_1 \le& 2\va + 2\vc + 3\vb + 3\vd 
\tag{Lemma \ref{lemma:pcsf3-ub}}\\
=& 2(\va + \vb + \vbb + \vc + \vd) + \vb - 2\vbb + \vd\\
\le& 2cost(\OPT) + (\vb - \vbb) - \vbb + \vd \tag{Lemma \ref{lemma:opt_lb}}\\
=& 2cost(\OPT) + \vba - \vbb + \vd \tag{Lemma \ref{lm:all_vb_cut_opt}}.
\end{align*}
\end{proof}

\begin{lemma}
\label{lemma:removeedge}
    For an instance $I$, it is possible to remove a set of edges from $\fopt$ with a total cost of at least $\vba$ while ensuring that the pairs in $\A$ remain connected.
\end{lemma}
\begin{proof}
Consider a single-edge set $S$ that cuts some pair $(i,j)$ in $\B$ with $\ysij>0$. 
Since $(i,j)$ is in $\B$, it is also in $\Q_1$ and therefore tight. 
By Lemma \ref{lemma:minimal-cut-tight}, any other pair cut by $S$ will also be tight. 
Consequently, the pairs in $\A$ will not be cut by $S$ since they are not tight.
Furthermore, according to Lemma \ref{lemma:single-cut}, if $S$ cuts only one edge $e$ of $\fopt$, then the only pairs that will be disconnected by removing edge $e$ from $\fopt$ are the pairs that are cut by $S$.
However, we have already shown that no pair in $\A$ is cut by $S$. 
Therefore, all pairs in $\A$ will remain connected even after removing edge $e$.
See Figure \ref{fig:remove-edges} for an illustration.

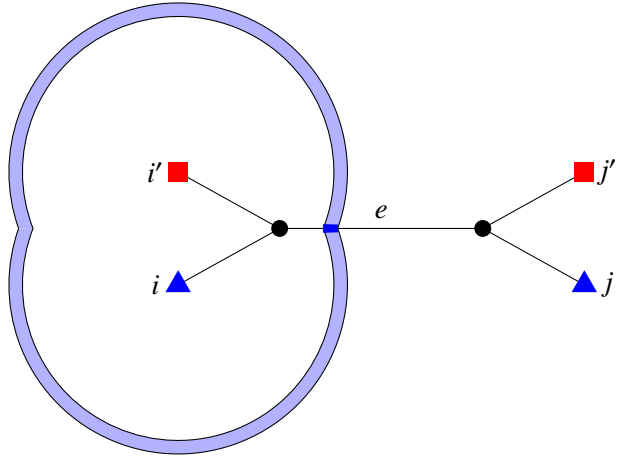
\begin{figure}
\centering
\begin{tikzpicture}[scale = 0.15]
\def\dx{9}
\def\dy{5}
\def\r{0.7}
\def\f{2}
\def\dif{1.2}
\def\d{19.5}
\foreach \x/\y/\R/\c in {-1/-1/15/-1, -1/1/15/1} {
    \draw[fill=blue!30] ([shift=(-\c*\d:\R)]\x*\dx, \y*\dy)  arc (-\c*\d:\c*(180+\d):\R);
    \draw[fill=white] ([shift=(-\c*\d-\c*\f:\R-\dif)]\x*\dx, \y*\dy)  arc (-\c*\d-\c*\f:\c*(180+\d+\f):\R-\dif);
}
\foreach \x/\y/\X/\Y/\t in {0/0/2/0/$e$, 0/0/-1/-1/, 0/0/-1/1/, 2/0/3/-1/, 2/0/3/1/} {
    \draw (\x*\dx, \y*\dy) -- node[above]{\t} (\X*\dx, \Y*\dy);
};
\foreach \x/\y/\R/\c in {-1/-1/15/-1, -1/1/15/1} {
    \draw[line width=1mm, blue]  ([shift=(-\c*\d:\R)]\x*\dx, \y*\dy) -- ([shift=(-\c*\d-\c*\f:\R-\dif)]\x*\dx, \y*\dy);
};
\foreach \x/\y in {0/0, 2/0} {
    \draw[fill=black]  (\x*\dx, \y*\dy) circle (\r);
};
\foreach \x/\y/\dir/\l in {-1/1/1/i', 3/1/-1/j'} {
    \def\rr{\r+0.5}
    \fill[fill=red] ([shift=(45:\rr)]\x*\dx, \y*\dy) -- ([shift=(135:\rr)]\x*\dx, \y*\dy) -- ([shift=(225:\rr)]\x*\dx, \y*\dy) -- ([shift=(315:\rr)]\x*\dx, \y*\dy);
    \node (a) at (\x * \dx - 2 * \dir, \y * \dy) {$\l$};
};
\foreach \x/\y/\dir/\l in {-1/-1/1/i, 3/-1/-1/j} {
    \def\rr{\r+0.6}
    \fill[fill=blue] ([shift=(-30:\rr)]\x*\dx, \y*\dy) -- ([shift=(-150:\rr)]\x*\dx, \y*\dy) -- ([shift=(-270:\rr)]\x*\dx, \y*\dy);
    \node (a) at (\x * \dx - 2 * \dir, \y * \dy) {$\l$};
};
\end{tikzpicture}
\caption{The figure shows the graph of $\fopt$ with pairs $(i, j)$ and $(i', j')$, and a single-edge set colored with pair $(i,j)$ in \dynamiccoloring. Tightness of $(i, j)$ implies tightness of $(i', j')$, and removing edge $e$ does not disconnect pairs in $\A$.} \label{fig:remove-edges} 
\end{figure}

For any single-edge set $S$ that cuts a pair $(i,j)$ in $\B$ with $\ysij>0$, we can safely remove the single edge of $\fopt$ that is cut by $S$.
The total amount of coloring on these removed edges is at least
\begin{align*}
    \sum_{S: \dopt(S)=1}\sum_{\substack{(i,j)\in\B
    \\S\odot(i,j)\\\ysij>0}} \ysij
    =\sum_{S: \dopt(S)=1}\sum_{\substack{(i,j)\in\B
    \\S\odot(i,j)}} \ysij = \vba.
\end{align*}
As the color on each edge does not exceed its length, the total length of the removed edges will also be at least $\vba$.
\end{proof}

Now, we introduce some useful notation to analyze the output of the recursive call. During the execution of $\rpcsf$ on an instance $\I$, it generates a modified instance $\R$ at Line \ref{line:construct_R}, where the penalties for pairs in $\Q_1$ are set to $0$. 
We use the notation $\pi'$ to represent the penalties in the instance $\R$ as they are defined in Lines~\ref{line:initial_pi}-\ref{line:set_pip}.
Since Line \ref{line:check_Q_1_is_empty} ensures that $\pi(\Q_1) \neq 0$, we can conclude that $\R$ is a reduced instance compared to $\I$, meaning that the number of pairs with non-zero penalties is smaller in $\R$ than in $\I$.
Given that we recursively call $\rpcsf$ on instance $\R$, we can bound the output of the recursive call by the optimal solution of $\R$ using induction. 
Let $\OPTR$ be an optimal solution for $\R$. 
We denote the forest of $\OPTR$ as $\foptr$ and the set of pairs not connected by $\foptr$ as $\qoptr$. 
The cost of $\OPTR$ is given by $cost(\OPTR) = \cc(\foptr) + \pi'(\qoptr)$.
We will use these notations in the following lemmas.
 
\begin{lemma}
\label{lemma:optr-up}
For an instance $\I$ and the instance $\R$ constructed at Line \ref{line:construct_R} during the execution of $\rpcsf(\I)$, we have 
$$cost(\OPTR) \le cost(\OPT) - \vd - \vba\text{.}$$ 
\end{lemma}
\begin{proof}
To prove this lemma, we first provide a solution for the instance $\R$ given the optimal solution of the instance $\I$, denoted as $\OPT$, and we show that the cost of this solution is at most $cost(\OPT) - \vd - \vba$. 
Since $\OPTR$ is a solution for the instance $\R$ with the minimum cost, we can conclude that $cost(\OPTR) \le cost(\OPT) - \vd - \vba$.

To provide the aforementioned solution for the instance $\R$, we start with the solution $\OPT$ consisting of the forest $\fopt$ and the set of pairs for which penalties were paid, denoted as $\qopt$. We create a new set $\qoptp = \qopt \cup \B =  \C \cup \D \cup \B$ and a forest $\foptp$ initially equal to $\fopt$. Since $\fopt$ connects pairs in $\A$ and $\B$, but we add pairs in $\B$ to $\qoptp$ and pay their penalties, we can remove edges from $\foptp$ that do not connect pairs in $\A$.

Let's focus on $\qoptp$ first. Since the penalties for pairs in $\B$ and $\D$ are set to $0$ in $\pi'$, we have
\begin{align*}
    \pi'(\qoptp) &= \pi'(\B) + \pi'(\C) + \pi'(\D) \tag{$\qoptp = \B \cup \C \cup \D$}\\
    &= \pi'(\C) \tag{$\pi'(\B)=\pi'(\D)=0$}\\
    &= \pi(\C) \\
    &= \pi(\qopt) - \pi(\D) \tag{$\qopt = \C \cup \D$}\\
    &= \pi(\qopt) - \sum_{(i, j) \in \D}\pij \\
    &= \pi(\qopt) - \sum_{(i, j) \in \D}\yij \tag{pairs in $\D$ are tight}\\
    &= \pi(\qopt) - \vd\text{.}
\end{align*}

Moreover, using Lemma \ref{lemma:removeedge}, we construct $\foptp$ from $\fopt$ by removing a set of edges with a total length of at least $\vba$, while ensuring that the remaining forest still connects all the pairs in $\A$. Therefore, we can bound the cost of $\foptp$ as
$$\cc(\foptp) \le \cc(\fopt) - \vba\text{.}$$

Summing it all together, we have
$$cost(\OPTR) \le \cc(\foptp) + \pi'(\qoptp) \le (\cc(\fopt) - \vba) + (\pi(\qopt) - \vd) = cost(\OPT) - \vd - \vba\text{,}$$
where the first inequality comes from the fact that $\OPTR$ is the optimal solution for the instance $\R$, while $(\qoptp, \foptp)$ gives a valid solution, i.e., $\foptp$ connects every pair that is not in $\qoptp$.
\end{proof}

Finally, we can bound the cost of the output of $\rpcsf$. For an instance $\I$, let's denote the cost of the output of $\rpcsf(\I)$ as $\costalg$.
In Theorem \ref{thm:main_theorem_with_proof}, we prove that the output of $\rpcsf$ is a 2-approximate solution for the PCSF problem.

\begin{theorem}
\label{thm:main_theorem_with_proof}
    For an instance $\I$, the output of $\rpcsf(\I)$ is a 2-approximate solution to the optimal solution for $\I$, meaning that 
    $$\costalg \le 2 cost(\OPT)\text{.}$$
\end{theorem}
\begin{proof}
    We will prove the claim by induction on the number of pairs $(i,j)$ with penalty $\pij>0$ in instance $\I$.

    First, the algorithm makes a call to the \hyperref[alg:pcsf3]{\pcsfthree} procedure to obtain a solution $(\Q_1,\Fp_1)$. 
    If $\pi(\Q_1)=0$ for this solution, which means no cost is incurred by paying penalties, the algorithm terminates and returns this solution at Line \ref{line:return-early}.
    This will always be the case in the base case of our induction where for all pairs $(i, j) \in \Q_1$, penalties $\pij$ are equal to $0$.
    Since every pair $(i, j) \in \Q_1$ is tight, we have $\yij=\pij=0$.
    Given that $\B$ and $\D$ are subsets of $\Q_1$, we can conclude that $\vb=\vba=\vbb=\vd=0$.
    Now, by Lemma \ref{lemma:pcsf3-ub-opt}, we have
    $$cost_1 \le 2cost(\OPT) + (\vba - \vbb) + \vd = 2cost(\OPT)\text{.}$$
    Therefore, when $\rpcsf$ returns at Line \ref{line:return-early}, we have
    $$\costalg = cost_1 \le 2cost(\OPT)\text{,}$$
    and we obtain a 2-approximation of the optimal solution.

    Now, let's assume that \pcsfthree~pays penalties for some pairs, i.e., $\pi(\Q_1)\neq 0 $. 
    Therefore, since we set the penalty of pairs in $\Q_1$ equal to $0$ for instance $\R$ at Line \ref{line:pip_set_zero}, the number of pairs with non-zero penalty in instance $\R$ is less than in instance $\I$.
    By induction, we know that the output of $\rpcsf$ on instance $\R$, denoted as $(\Q_2, \Fp_2)$, has a cost of at most $2cost(\OPTR)$.
    That means
    $$\cc(\Fp_2) + \pi'(\Q_2) \le 2cost(\OPTR)\text{.}$$
    In addition, we have 
    \begin{align*}
        \pi(\Q_2) = \pi(\Q_2 \setminus \Q_1) + \pi(\Q_2 \cap \Q_1)
        \le \pi'(\Q_2 \setminus \Q_1) + \pi(\Q_1)
        \le \pi'(\Q_2) + \pi(\Q_1)\text{,}
    \end{align*}
    where we use the fact that $\pi'_{ij}=\pij$ for $(i, j) \notin \Q_1$.
    Now we can bound the cost of the solution $(\Q_2, \Fp_2)$, denoted as $cost_2$, by
    \begin{align*}
        cost_2 &= \cc(\Fp_2) + \pi(\Q_2) \\
        &\le \cc(\Fp_2) + \pi'(\Q_2) + \pi(\Q_1) \\
        &\le 2cost(\OPTR) + \pi(\Q_1) \tag{By induction}\\
        &\le 2\left(cost(\OPT) - \vd - \vba\right) + \sum_{(i, j) \in \Q_1} \pij \tag{Lemma \ref{lemma:optr-up}}\\
        &= 2\left(cost(\OPT) - \vd - \vba\right) + \sum_{(i, j) \in \Q_1} \yij \tag{pairs in $\Q_1$ are tight}\\
        &= 2cost(\OPT) - 2\vd - 2\vba + \vb + \vd \\
        &= 2cost(\OPT) - \vba + \vbb - \vd\text{.} \tag{Lemma \ref{lm:all_vb_cut_opt}}
    \end{align*}
    
    Furthermore, according to Lemma \ref{lemma:pcsf3-ub-opt}, the cost of the solution $(\Q_1, \Fp_1)$, denoted as $cost_1$, can be bounded by
    \begin{align*}
    cost_1 \le 2\OPT + \vba - \vbb + \vd \text{.}
    \end{align*}

    Finally, in Line \ref{line:return-min}, we return the solution  with the smaller cost between $(\Q_1, \Fp_1)$ and $(\Q_2, \Fp_2)$. Based on the upper bounds above on both solutions, we know that
    \begin{align*}
        \costalg &= \min(cost_1, cost_2) \le \frac{1}{2}(cost_1 + cost_2) \\ 
        &\le \frac{1}{2} (2cost(\OPT) + \vba - \vbb + \vd + 2cost(\OPT) - \vba + \vbb - \vd) \\
        &= \frac{1}{2} (4cost(\OPT)) = 2cost(\OPT)\text{,}
    \end{align*}
    and we obtain a 2-approximation of the optimal solution.
    This completes the induction step and the proof of the theorem.
\end{proof}

\begin{theorem}
    The runtime of the \rpcsf~algorithm is polynomial.
\end{theorem}
\begin{proof}
    Let $n$ be the number of vertices in the input graph. There are $O(n^2)$ pairs of vertices in total. 
    Whenever \rpcsf~calls itself recursively, the number of pairs with non-zero penalties decreases by at least one, otherwise \rpcsf~will return at Line \ref{line:return-early}. 
    Thus, the recursion depth is polynomial in $n$. 
    At each recursion level, the algorithm only runs \pcsfthree~on one instance of the problem and performs $O(n^2)$ additional operations. 
    By Lemma \ref{lm:pcsfthree_polynomial}, we know that \pcsfthree~runs in polynomial time.
    Therefore, the total run-time of \rpcsf~will also be polynomial.
\end{proof}

\subsection{Improving the approximation ratio}
\label{sec:3-2-n-approax}
In this section, we briefly explain how a tighter analysis can be used to show that the approximation ratio of the \rpcsf~algorithm is at most $2-\frac{1}{n}$, where $n$ is the number of vertices in the input graph $G$. This approximation ratio more closely matches the approximation ratio of $2-\frac{2}{n}$ for the Steiner Forest problem.

We first introduce an improved version of Lemmas \ref{lemma:pcsf3-ub} and \ref{lemma:pcsf3-ub-opt}.
\begin{lemma}
    \label{lemma:pcsf3-ub-tight}
    For an instance $\I$, during the first iteration of $\rpcsf(\I)$ where $\pcsfthree(\I)$ is invoked, we have the following upper bound
    $$cost_1 \le (2-\frac{2}{n})\cdot\va + (2-\frac{2}{n})\cdot\vc + (3-\frac{2}{n})\cdot\vb + (3-\frac{2}{n})\cdot\vd\text{.}$$
\end{lemma}
\begin{proof}
    We proceed similarly to the proof of Lemma \ref{lemma:pcsf3-ub} and make a slight change. In one of the last steps of that proof, we use the following inequality:
    $$
    \sum_{v\in V_a\cup V_i} d_H(v) - \sum_{v\in V_i} d_H(v) \\
        \le 2(\lvert V_a \rvert + \lvert V_i\rvert)-\sum_{v\in V_i} d_H(v).
    $$
    This is true, as $H$ is a forest and its number of edges is less than its number of vertices. However, as the number of edges in a forest is strictly less than the number of vertices, we can lower the right-hand side of this inequality to $2(\lvert V_a \rvert + \lvert V_i\rvert -1)-\sum_{v\in V_i} d_H(v)$. Rewriting the main inequality in this step with this change gives us
    \begin{align*}
        \sum_{S\in\activesets} d_{\Fp_1}(S) 
        &\le 2(\lvert V_a \rvert + \lvert V_i\rvert-1)-\sum_{v\in V_i} d_H(v) \\
        &\le 2(\lvert V_a \rvert + \lvert V_i\rvert-1)- 2\lvert V_i \rvert \tag{$d_H(v)\ge 2$ for $v\in V_i$}\\
        &\le 2(\lvert V_a\rvert-1) = 2\lvert\activesets\rvert -2\tag{$\lvert V_a\rvert = \lvert\activesets\rvert$}\\
        &=(2-\frac{2}{\lvert\activesets\rvert})\lvert\activesets\rvert \\&\le (2-\frac{2}{n})\lvert\activesets\rvert \tag{$\lvert\activesets\rvert \le n$}\text{.}
\end{align*}
Based on the steps in the proof of Lemma \ref{lemma:pcsf3-ub}, this leads to the desired upper bound.
\end{proof}
\begin{lemma}
\label{lemma:pcsf3-ub-opt-tight}
    For an instance $\I$, during the first iteration of $\rpcsf(\I)$ where $\pcsfthree(\I)$ is invoked, we can establish an upper bound on the output of $\pcsfthree$ as follows:
    $$cost_1 \le (2-\frac{2}{n})\cdot cost(\OPT) + \vba - (1-\frac{2}{n})\cdot\vbb + \vd$$
\end{lemma}
\begin{proof}
We prove this lemma similarly to Lemma \ref{lemma:pcsf3-ub-opt}, except we use Lemma \ref{lemma:pcsf3-ub-tight} instead of Lemma \ref{lemma:pcsf3-ub}.
    \begin{align*}
    cost_1 \le& (2-\frac{2}{n})\cdot\va + (2-\frac{2}{n})\cdot\vc + (3-\frac{2}{n})\cdot\vb + (3-\frac{2}{n})\cdot\vd 
    \tag{Lemma \ref{lemma:pcsf3-ub-tight}}\\
    =& (2-\frac{2}{n})(\va + \vb + \vbb + \vc + \vd) + \vb - (2-\frac{2}{n})\cdot\vbb + \vd\\
    \le& (2-\frac{2}{n})\cdot cost(\OPT) + (\vb - \vbb) - (1-\frac{2}{n})\vbb + \vd \tag{Lemma \ref{lemma:opt_lb}}\\
    =& (2-\frac{2}{n})\cdot cost(\OPT) + \vba - (1-\frac{2}{n})\cdot \vbb + \vd \tag{Lemma \ref{lm:all_vb_cut_opt}}.
    \end{align*}
\end{proof}

Finally, we improve Theorem \ref{thm:main_theorem_with_proof}.

\begin{theorem}
        For an instance $\I$, the output of $\rpcsf(\I)$ is a $(2-\frac{1}{n})$-approximate solution to the optimal solution for $\I$, meaning that 
    $$\costalg \le (2-\frac{1}{n})\cdot cost(\OPT)\text{.}$$
\end{theorem}
\begin{proof}
    Similarly to the proof of Theorem \ref{thm:main_theorem_with_proof}, we use induction on the number of non-zero penalties. If the algorithm terminates on Line \ref{line:return-early} then by Lemma \ref{lemma:pcsf3-ub-opt-tight} we have
    $$cost_1 \le (2-\frac{2}{n})\cdot cost(\OPT) + \vba - (1-\frac{2}{n})\cdot\vbb + \vd = (2-\frac{2}{n})\cdot cost(\OPT)$$
    since $\vba$, $\vbb$, and $\vd$ are all $0$ in this case. As $2-\frac{2}{n}\le 2-\frac{1}{n}$, the desired inequality holds in this case. This establishes our base case for the induction.

    Using the same reasoning as the proof of Theorem \ref{thm:main_theorem_with_proof}, based on the induction we have 
    \begin{align*}
    cost_2 &\leq (2-\frac{1}{n})\cdot cost(\OPTR) + \pi(\Q_1) \\
    &= (2-\frac{1}{n})\cdot cost(\OPTR) + \vb + \vd    \\
    &\le (2 - \frac{1}{n}) (cost(\OPT) - \vba - \vd) + \vb + \vd \tag{By Lemma \ref{lemma:optr-up}}\\
    & \le (2 - \frac{1}{n})\cdot cost(\OPT) - (1 - \frac{1}{n})\cdot \vba + \vbb - (1 - \frac{1}{n}) \cdot\vd\text{.}
    \end{align*}
    We can combine this with the following upper bound from Lemma \ref{lemma:pcsf3-ub-opt-tight}
    $$cost_1 \le (2-\frac{2}{n})\cdot cost(\OPT) + \vba - (1-\frac{2}{n})\cdot\vbb + \vd\text{.}$$
    As the algorithm chooses the solution with the lower cost between $cost_1$ and $cost_2$, we have
    \begin{align*}
        \costalg &= \min(cost_1, cost_2) \le \frac{1}{2}(cost_1 + cost_2) \\ 
        &\le \frac{1}{2} \left [ (2-\frac{2}{n})\cdot cost(\OPT) + \vba - (1-\frac{2}{n})\cdot\vbb + \vd \right. \\& \left.+ (2-\frac{1}{n})\cdot cost(\OPT) - (1-\frac{1}{n})\cdot \vba + \vbb - (1-\frac{1}{n})\cdot \vd \right ] \\
        &= \frac{1}{2} \left((4-\frac{3}{n})\cdot cost(\OPT) + \frac{2}{n}\vbb + \frac{1}{n} \vba + \frac{1}{n}\vd\right)\\
        &\le \frac{1}{2} \left((4-\frac{2}{n})\cdot cost(\OPT) + \frac{1}{n}[2\vbb + \vba + \vd -cost(\OPT)]\right)\\
        &\le \frac{1}{2}(4-\frac{2}{n})\cdot cost(\OPT) \tag{$cost(\OPT)\geq 2\vbb + \vba + \vd$ by Lemma \ref{lemma:opt_lb}}\\
        &= (2-\frac{1}{n})\cdot cost(\OPT)\text{.}
    \end{align*}
    Therefore, the algorithm obtains a $(2-\frac{1}{n})$-approximation of the optimal solution.  
\end{proof}

\section{Acknowledgements}

The work is partially support by DARPA QuICC, NSF AF:Small  \#2218678, and  NSF AF:Small  \#2114269


\bibliographystyle{abbrv}
\bibliography{references}

\end{document}